%% file: main.tex
\documentclass[final,12pt]{colt2026} % Include author names

\title[Space-Efficient Language Generation in the Limit]{Space-Efficient Language Generation in the Limit}

\usepackage{macros}
\usepackage{times}

\newif\ifarxivversion
\arxivversiontrue
\ifarxivversion
  \makeatletter
  \editor{}
  \jmlrproceedings{}{}
  \jmlrvolume{}
  \jmlrpages{}
  \jmlryear{}
  \jmlrworkshop{}
  \renewcommand*{\@titlefoot}{}
  \makeatother
\fi

\coltauthor{%
 \Name{Nicolas Flammarion} \Email{nicolas.flammarion@epfl.ch}\\
 \addr EPFL, Switzerland
 \AND
 \Name{Chirag Pabbaraju} \Email{cpabbara@stanford.edu}\\
 \addr Stanford University, USA
 \AND
 \Name{Hristo Papazov} \Email{hristo.papazov@epfl.ch}\\
 \addr EPFL, Switzerland
 \AND
 \Name{Miltiadis Stouras} \Email{miltiadis.stouras@epfl.ch}\\
 \addr EPFL, Switzerland
 \AND
 \Name{Ola Svensson} \Email{ola.svensson@epfl.ch}\\
 \addr EPFL, Switzerland
}

\begin{document}

\maketitle
% \vspace{-10mm}
\begin{abstract}
We initiate a resource-aware theory of \textit{language generation in the limit} under the minimal constraint of space efficiency. In our framework, a learner observes an adversarial positive stream from a target language $K$ and must eventually output a hallucination-free hypothesis language $L \subseteq K$ while omitting at most $\Delta$ strings of $K$. We focus on $\mathcal{C}_{s,k}$, the collection of languages recognized by DFAs with at most $s$ states over an alphabet of size $k$, as the natural hypothesis class for memory-bounded learners. In the exponential-space regime, we prove that a learner can exactly identify the target $K$. Under a stricter memory budget, we characterize the strongest possible generation guarantees. In particular, we present a streaming algorithm using $\mathrm{poly}(s,k)$ space that converges to a hypothesis with generation gap $\Delta = O(k^{2s-2})$. Moreover, the learned hypothesis captures every string in $K$ of length at least $2s-1$. We complement this result with a near-matching lower bound through a reduction from a standard communication complexity problem. Specifically, achieving generation gap $\Delta \le k^{(1-\varepsilon)s}$ requires $k^{\Omega(\varepsilon s)}$ memory. Together, these results reveal a sharp transition between polynomial-space generation and exponential-space exact identification.
\end{abstract}

\begin{keywords}
  Language Generation in the Limit, Finite-State Automata, Computational Efficiency
\end{keywords}
\input{Sections/Introduction}
\input{Sections/Formal_Framework}
\input{Sections/Upper_Bounds}
\input{Sections/Lower_Bounds}
\input{Sections/Conclusion}

\newpage
% Acknowledgments---will not appear in anonymized version
\acks{
This work was partially funded by the Swiss National Science Foundation, grant number 212111.
Chirag Pabbaraju is supported by Gregory Valiant's and Moses Charikar's Simons Investigator Awards, and a Google PhD Fellowship.
Miltiadis Stouras and Ola Svensson are supported by the Swiss State Secretariat for Education, Research and Innovation (SERI) under contract number MB22.00054.
}

\bibliography{references}

%%%%%%%%%%%%%%%%%%%%%%%%%%%%%%%%%%%%%%%%%%%%%%%%%%
%%%%%%%%%%%%%%%%%%%%%%%%%%%%%%%%%%%%%%%%%%%%%%%%%%
\newpage 
\appendix
\crefalias{section}{appendix} % uncomment if you are using cleveref

\section*{Organization of the Appendix}
\begin{itemize}[label=\ding{226}, topsep=5pt, left=10pt, itemsep=0pt, labelsep=10pt]
    \item Appendix~\ref{app:communicative} motivates the focus on finite regular hypothesis spaces by developing the connection to communicative languages induced by interaction among space-bounded agents.
    \item Appendix~\ref{sec:appendix-dfa-facts} provides standard results on DFAs.
    \item Appendix~\ref{sec:appendix-upper-bounds} contains deferred proofs from \Cref{sec:upper-bounds}.
    %presents the streaming algorithm, proves convergence under $\mathrm{poly}(s,k)$ space, and establishes the bounded-gap guarantee $\Delta = O(k^{2s-2})$, including the stronger property that the algorithm captures every string in $K$ of length at least $2s-1$.
    \item Appendix~\ref{sec:appendix-lower-bounds} contains deferred proofs from \Cref{sec:lower-bounds}.
    %gives the near-matching memory lower bound via a reduction from a standard communication-complexity problem, showing that any algorithm achieving generation gap $\Delta \le k^{(1-\varepsilon)s}$ requires $k^{\Omega(\varepsilon s)}$ memory.
\end{itemize}

\input{Appendix_Sections/Communicative_Languages}
\input{Appendix_Sections/DFA_Stuff}

\input{Appendix_Sections/Appendix_Upper_Bounds}
\input{Appendix_Sections/Appendix_Lower_Bounds}

\end{document}

%% file: Sections/Introduction.tex
\section{Introduction}
\label{sec:introduction}
Large language models (LLMs) have been shown to generate novel and grammatically well-formed text even after training on exclusively \textit{positive} natural-language examples \citep{radford2019language, brown2020language, mahowald2024dissociating}. This remarkable empirical success revives a classical question from computational learning theory: What kind of language acquisition becomes possible when a learner only observes valid strings from the target language with no explicit negative feedback?

\textbf{Identification in the Limit.} A canonical formalization of this positive-only setting comes from Gold's work on learning in the limit. Motivated by the now-contentious\footnote{Consider \citep{marcus1993negative} for arguments in favor of language acquisition solely from positive examples and see \citep{chouinard2003adult} for arguments for the presence of negative examples in parent-infant communication.} claim from psycholinguistics that infants learn the grammar of a language solely from positive examples, \citet{gold1967language} developed the learning-in-the-limit framework as a minimalist formalization of language learning from positive data. Gold’s model fixes (i) a representation system for languages, (ii) a target language, and (iii) an infinite stream that eventually lists every string in the target language. After each received valid string from this language presentation, a learner outputs a hypothesis language. Gold defined a language family as \textit{identifiable in the limit} when a learner eventually stabilizes on a hypothesis that matches the ground-truth exactly, regardless of the target language and the presentation order. Unfortunately, under this learning criterion, \citet{gold1967language} proved that only highly restricted\footnote{See \citep[Theorem 1]{angluin1980inductive} for a precise characterization of identifiable collections in the limit.} language collections are learnable. In particular, one cannot even learn the class of regular languages. These impossibility results largely redirected mainstream learning theory toward distributional, sample-complexity, and efficiency guarantees, crystallized in PAC learning \citep{valiant1984theory,kearns1994introduction}.
% while identification in the limit continued mainly within grammatical inference \citep{jain1999systems,de2010grammatical}.

\textbf{Generation in the Limit.} Only recently, \citet{kleinberg2024language} revitalized Gold’s positive presentation framework by proposing a different language-acquisition criterion that fits language modeling more closely. Specifically, the learner still receives an arbitrary enumeration of the target language, but success no longer requires eventual equality between hypothesis and target. Instead, after some finite time, the learner must begin to generate infinitely many unseen strings from the target language. Put informally, sufficient exposure to positive data should make hallucinations disappear while preserving an infinite learned subset of the target language. This shift replaces identification with generation by relaxing the objective from equality to containment. Surprisingly, the relaxed objective substantially alters the theory: Now, under oracle access to membership queries for the enumerated language family, generation in the limit becomes achievable for every countable collection of languages.\footnote{Of course, if we insist on computable membership queries, then the scope of the generatable collections decreases.}
\citet{kleinberg2024language}'s positive result has since inspired a growing line of follow-up work examining different aspects of generation/learning in the limit \citep{kalavasis2024characterizations, li2024generation, charikar2024exploring, papazov2025learning, charikar2025pareto, peale2025representative, raman2025generation, hanneke2025union, charikar2025characterization}. Interestingly, the current generation-in-the-limit literature largely treats computation as free. As a result, generation guarantees rest on models that place no bounds on computation and, in particular, allow unbounded memory. Such assumptions diverge from practice, where human learners operate without external storage and language models train and run under explicit memory budgets.

\textbf{Space-Efficient Generation.} 
Motivated by this explanatory gap, our paper initiates a resource-aware study of generation in the limit by imposing the minimal computational restriction of space efficiency. Under our proposed framework, (i) the learner receives a stream of positive examples through an online interface, (ii) processes each target string symbol-by-symbol while operating within a strict memory budget, and (iii) after each processed example, outputs a hypothesis representation that must eventually generate a subset of the target language. Now, the requirement of bounded memory implies that our learning algorithms can only occupy finitely many internal configurations. Therefore, membership verification for the streamed examples can only range over regular languages, and only over a finite set of such languages. Accordingly, space-efficient generation only admits a meaningful formulation over finite collections of regular languages. Despite this restriction, regular-language generation remains practically significant. Specifically, in \Cref{app:communicative}, we argue that the collection of \textit{communicative languages}, which arise from the interaction of space-bounded agents such as humans \citep{miller1963finitary}, forms a strict subset of regular languages.

Concurrently and independently, \citet{kleinberg2026languagegenerationlimitbounded} also study generation in the limit under bounded-memory restrictions. Their results are incomparable to ours: they consider broader language collections with coarser memory models, where the learner can only retain a limited window or buffer of past examples. In contrast, we impose an explicit bit-space budget on the learners and focus on regular languages (\(C_{s,k}\)). This lets us prove quantitative space--breadth tradeoffs, including tight bounds
on the generation gap achievable with polynomial space.

\textbf{Overview and Contributions.}
The regularity arguments above justify restricting attention to regular hypothesis spaces parameterized by automaton complexity. For a fixed alphabet $\Sigma$ with $|\Sigma|=k$ and a state bound $s$, let $\mathcal{C}_{s,k}$ denote the collection of languages recognizable by deterministic finite automata (DFAs) over $\Sigma$ with at most $s$ states. The remainder of the paper characterizes space-efficient generation within $\mathcal{C}_{s,k}$. Upper bounds give explicit streaming learners operating in $\p(s,k)$ space, while lower bounds show that stronger generation guarantees require exponential memory. These results expose a sharp tradeoff between memory and generative breadth, yielding a resource-sensitive theory of generation in the limit.

\textbf{Roadmap.}
Section~\ref{sec:framework-main-results} formalizes the space-efficient generation model and describes our main results. Section~\ref{sec:upper-bounds} presents a space-efficient generation algorithm for $\mathcal{C}_{s,k}$.
Section~\ref{sec:lower-bounds} proves matching lower bounds via communication complexity reductions.
Section~\ref{sec:conclusion} concludes with implications for resource-bounded language acquisition and open directions.

%% file: Sections/Formal_Framework.tex
\section{Space-Efficiency Framework and Main Results} \label{sec:framework-main-results}

As described in \Cref{sec:introduction}, we study the problem of space-efficient generation in the limit where an adversary enumerates $w_1, w_2, w_3, \dots$ and presents a target language $K \in \SS$ in a \textit{streaming} fashion. Each string $w_t$ arrives one symbol at a time, with a designated delimiter separating consecutive strings. After each symbol, the learning algorithm $\A$ performs computation under a polynomial space budget, and after each processed string $w_t \in K$, $\A$ outputs a hypothesis representation of a regular language. 

Concretely, for the rest of the paper, we fix an alphabet $\Sigma$ of size $k$ and let $\mathcal{C}_{s,k}$ denote the collection of all regular languages over $\Sigma$ recognizable by DFAs of at most $s$ states. As a quick reminder, we recall the textbook definition of a DFA \citep{sipser1996introduction,hopcroft2001introduction}.
\begin{definition}[Deterministic Finite-State Automaton] \label{def:dfa}
    A DFA $A$ over the finite alphabet $\Sigma$ is a 5-tuple $A = (Q, \Sigma,\d, q_0, F)$, where $Q$ is a finite set of states, $\d : Q \times \Sigma \to Q$ is a transition function, $q_0 \in Q$ is the initial state, and $F \subseteq Q$ is a set of accepting (final) states. We extend the transition function $\d$ to act on $\SS$ by setting $\d(q,\eps)=q$ (where $\eps$ is the empty string) and recursively defining $\d(q, \s w) \coloneqq \d(\d(q, \s), w)$, for every $q \in Q, \sigma \in \Sigma$ and $w \in \SS$. We use the notation $L(A)$ to denote the regular language accepted by the DFA $A$, i.e., $L(A)=\{w \in \SS:\delta(q_0,w) \in F\}$, and we define the size of the DFA $A$ as $|A| = |Q|$.
\end{definition}
We now state our main definition of space-efficient language generation in the limit.
\begin{definition}[Space-Efficient Generation in the Limit]
    \label{def:space-efficient-generation-in-the-limit}
    A space-efficient algorithm $\mathcal{A}$ generates in the limit from $\mathcal{C}_{s,k}$ with a generation gap $\Delta_\mathcal{A}(s,k)$ if (i) $\mathcal{A}$ uses at most $\p(s,k)$ bits of working memory, and (ii) for any target regular language $K \in \mathcal{C}_{s,k}$ and any surjective enumeration $w : \N \surj K$ presented to $\mathcal{A}$ in a streaming fashion, there exists a finite time $t^\star$ such that for all $t \ge t^\star$, $\mathcal{A}$ outputs a representation of a DFA $\mathcal{A}(t)$\footnote{We use $\A(t)$ as a shorthand for $\A(w_1, \dots, w_t)$.} satisfying $L(\mathcal{A}(t)) \subseteq K$
    \footnote{One can also formulate the objective without the hard no-hallucination requirement by asking directly for a bounded symmetric difference $|L(\A(t)) \triangle K| \le \Delta_\mathcal{A}(s,k)$. The upper and lower bounds in this paper remain the same under that relaxed formulation as we show in \Cref{sec:upper-bounds,sec:lower-bounds}.}
    and $|K \setminus L(\mathcal{A}(t))| \le \Delta_\mathcal{A}(s,k)$.
\end{definition}
The literature refers to the requirement of outputting a language representation satisfying $L(\A(t)) \subseteq K$ as \textit{index-based generation in the limit} \citep{kleinberg2025density}. While index-based generation with an unbounded generation gap admits a trivial solution, for example by repeatedly outputting the singleton language ${w}$ for some short string $w \in K$, such a strategy fails to approximate $K$ in any rich or informative sense. In particular, constructing an infinite sublanguage of $K$ already poses a nontrivial challenge. This paper studies a substantially stronger objective: generation in the limit under a bounded generation gap, where the output language  differs from $K$ on only finitely many strings, thereby capturing almost the entire target language.

Without an \textit{a priori} bound $s$ on the size of a DFA recognizing the target language, bounded-gap generation in the limit remains impossible, even without memory constraints. More formally, \citet[Theorem 3.10]{kalavasis2024characterizations} show that any collection admitting bounded-gap\footnote{\citet{kalavasis2024characterizations} refer to this notion as generation with ``approximate breadth''.} generation in the limit must satisfy the so-called \textit{Weak Angluin's Condition}. In particular, regular languages violate the weak Angluin's condition. Consequently, bounded-gap generation in the limit fails over the full class of regular languages (without any \textit{a priori} state bound), even with unlimited memory.

Interestingly, \citet{angluin1980inductive} proved that any finite collection of languages allows identification in the limit. Hence, without the polynomial-space requirement, we can completely remove the generation gap and \textit{identify} $\mathcal{C}_{s,k}$ in the limit. We describe such an algorithm in \Cref{sec:upper-bounds}; here, we simply note that identification crucially relies on using exponential space.
Then, how does the picture change with polynomial space constraints? Perhaps surprisingly, our first result shows that there exists a space-efficient algorithm that generates from $\mathcal{C}_{s,k}$ in the limit with a bounded generation gap.
\vspace{-2mm}
\begin{theorem}[Space-Efficient Generation in the Limit]
    \label{thm:intro-main-upper-bound}
    There exists a space-efficient learning algorithm $\mathcal{A}$ that generates in the limit from $\mathcal{C}_{s,k}$ with a generation gap $\Delta_\mathcal{A}(s,k) \le O(k^{2s-2})$.
\end{theorem}
\vspace{-2mm}
In fact, the algorithm satisfies a stronger guarantee: The output language omits only target strings of length at most $2s-2$. At a high-level, our algorithm traverses $\mathcal{C}_{s,k}$ according to a predefined topological order and outputs the first language not proven inconsistent with the most recently observed input string. At the next input, the learner continues traversing $\mathcal{C}_{s,k}$ starting from the previously outputted language. The topological order ensures that the algorithm converges to a language with a finite symmetric difference with the target language. We make the traversal space-efficient by using a recursion technique inspired by Savitch's famous theorem in complexity theory \citep{savitch1970relationships}.
Finally, to achieve generation in the limit while minimizing the generation gap, we invoke results from the automata-minimization literature \citep{gawrychowski2011minimising} that bound the size of the finite symmetric difference between the learner's hypothesis and the target.

Our space-efficient algorithm nonetheless incurs a nonzero generation gap of $O(k^{{2s-2}})$. This bound raises a natural question: Does polynomial-space learning permit a zero generation gap, or, equivalently, identification in the limit over $\C_{s,k}$? The next theorem rules out this possibility. In fact, we prove something much stronger: Any learner that operates with sub-exponential space must incur an exponential generation gap.
\begin{theorem}[Space--Breadth
Tradeoff]
    \label{thm:intro-main-lower-bound}
    For any $\varepsilon > 0$, any algorithm $\mathcal{A}$ generating in the limit from $\mathcal{C}_{s,k}$ with a generation gap $\Delta_{\mathcal{A}}(s,k) \leq k^{(1-\varepsilon)s}$ must use $k^{\Omega(\varepsilon s)}$ memory bits. Moreover, any algorithm $\A$ achieving the relaxed condition $|L(\A(t^\star)) \triangle K| \leq k ^{(1-\eps)s}$ must still use $k^{\Omega(\varepsilon s)}$ bits.
\end{theorem}
Most existing lower bounds in the identification/generation-in-the-limit literature involve some form of a diagonalization argument: The adversary constructs an enumeration of a problematic language by endlessly switching back-and-forth between enumerating different languages and causing the algorithm to fail at the objective at each switch in the process. These diagonalization arguments crucially rely on a special nesting structure present among the languages in the collection.

In contrast, our lower bound follows from a reduction to the \emph{Index} problem from communication complexity: a standard route to space lower bounds for streaming algorithms. The Index problem defines a one-way communication game between Alice and Bob: Alice receives a length-$n$ vector $x \in [k]^n$, Bob receives an index $i \in \{1,\ldots,n\}$, and Bob must output $x_i$ after receiving a single message from Alice. A standard pigeonhole-principle bound implies that any protocol deterministically solving Index requires communication of $n \log_2 k$ bits from Alice to Bob\footnote{Related strong lower bounds, with quantitatively similar rates, also hold for randomized one-way communication protocols; see, for example, \cite{bar-yossef2002coco,jayram2008toc}.}. To prove the streaming lower bound, we show that a space-efficient learner using only $\p(s,k)$ bits of memory and achieving sufficiently small generation gap induces a one-way protocol for Index with communication strictly below $n \log_2 k$ bits, yielding a contradiction.

The lower bound also shows that the generation gap achieved by the algorithm matches the optimal rate up to constant factors in the exponent. Together, \Cref{thm:intro-main-upper-bound} and \Cref{thm:intro-main-lower-bound} establish that the classical separation between identification and generation in the limit persists under polynomial space bounds.
This separation vanishes once the space budget increases from $\p(s,k)$ to $\exp(s,k)$. Under $\exp(s,k)$ space, an algorithm can identify $\mathcal{C}_{s,k}$ in the limit. Space-bounded generation in the limit therefore exhibits a sharp phase transition: Sub-exponential space forces an exponential generation gap, while exponential space permits exact identification.

%% file: Sections/Upper_Bounds.tex
\section{Upper Bounds}
\label{sec:upper-bounds}
In this section, we derive our space-efficient algorithm $\A$, which, for all sufficiently large $t$, outputs a DFA $\A(t)$ that satisfies $L(\A(t)) \subseteq K$, where $K$ denotes the unknown target language assumed to require a DFA of size at most $s$ for recognition. Furthermore, $\A(t)$ satisfies that $|K \setminus L(\A(t))|\le O(k^{2s-2})$. In particular, $L(\A(t))$ contains all strings in $K$ that have length at least $2s-1$. Before describing our learning procedure, we first recall Angluin's algorithm for identification in the limit to illustrate the challenges for achieving space-efficiency.

\paragraph{Recalling Angluin's identification algorithm for finite collections.} Let us arbitrarily enumerate the finitely many DFAs of size at most $s$ as $A_1, A_2, \dots,A_N$.\footnote{Our enumeration will contain many isomorphic DFAs recognizing the same language. This redundancy is harmless. We only require that each regular language in $\mathcal{C}_{s,k}$ receives some representation in the enumeration.} Here, $N \le 2^s \cdot s^{ks}$ ($2^s$ ways of selecting an accepting subset of at most $s$ states, followed by specifying the transition on each symbol for every state). By assumption, there exists some index $i^\star$ for which $L(A_{i^\star})=K$. As established by \citet{angluin1980inductive}, all finite collections are identifiable in the limit (without space constraints). The algorithm operates as follows: It first topologically sorts the DFAs $A_1, A_2, \dots,A_N$ according to the \textit{strict partial order} $\prec_1$ induced by \textit{strict inclusion}: 
\begin{equation}
    \label{eqn:po1-def}
    A \prec_1 A' \iff L(A) \subset L(A').
\end{equation}
Namely, if $L(A)$ is a strict subset of $L(A')$, then $A$ appears before $A'$ in the topological sort.\footnote{It can be readily verified that the relation in \eqref{eqn:po1-def} is a strict partial order: It is irreflexive, asymmetric and transitive.} Indeed, since the collection is finite, a topological sort corresponding to any partial order exists, and can be constructed by a standard depth-first search over a directed graph whose nodes correspond to the DFAs $A_i$.
So, abusing notation slightly, suppose that the DFAs $A_1,\dots,A_N$ are already ordered according to the topological sort induced by $\prec_1$. 
At time step $t$, after having seen input $w_1,\dots,w_t$, the algorithm outputs the DFA $A_{i(t)}$, where $i(t)$ is the smallest index of a DFA that accepts each of $w_1,\dots,w_t$\footnote{Such an $i(t)$ always exists, since the target language $K$ is recognized by some DFA in the collection.}. Concretely, for every $j < i(t)$, it holds that $\{w_1,\dots,w_t\} \nsubseteq L(A_j)$. Let $i^\star$ be the index of the DFA that recognizes $K$. Note that $\{w_1,\dots,w_t\} \subseteq L(A_{i^\star})$ for every $t$. Therefore, $i(t) \le i^\star$ for every $t$; since $i(t)$ is non-decreasing in $t$, $i(t)$ must eventually converge. Furthermore, the DFA $A_{i(t)}$ at which the algorithm converges must satisfy $L(A_{i^\star}) \subseteq L(A_{i(t)})$, since otherwise, there exists some $w \in L(A_{i^\star}) \setminus L(A_{i(t)})$, which is guaranteed to show up eventually in the input, causing $A_{i(t)}$ to be invalidated. Finally, since $i(t) \le i^\star$, and $L(A_{i^\star}) \subseteq L(A_{i(t)})$, it must necessarily be the case that $L(A_{i^\star}) = L(A_{i(t)})$, by the definition of $\prec_1$.

An important feature of the algorithm above is that it keeps track of the entire history of inputs seen so far at every time step. This is not feasible in our framework of space efficiency where we insist on maintaining only $\p(s, k)$ bits of working memory across time steps. As we shall see later, it is actually sufficient to remember only those input strings that have length at most $O(s)$ for exact identification; however, even this requires $\exp(s,k)$ memory. 

There is also the separate issue of traversing the DFAs $A_1,\dots,A_N$ in an order topologically sorted according to $\prec_1$ in a space-efficient manner. We address this latter issue first.
For any automata $C,D$, checking $C \prec_1 D$ amounts to testing $L(C) \setminus L(D)=\emptyset$ and $L(D)\setminus L(C) \neq \emptyset$. Both conditions can be checked via the standard product constructions (see \Cref{sec:appendix-dfa-facts}) of DFAs (on at most $s^2$ states) that recognize $L(C) \setminus L(D)$ and $L(D) \setminus L(C)$, which can be done in $\p(s, k)$ time (and hence also space). Thereafter, testing whether the product automaton recognizes an empty language amounts to checking if there exists a directed path beginning from the initial state and reaching an accepting state, which can again be done in $\p(s,k)$ time (and hence space). Thus, one can verify $C \prec_1 D$ in $\p(s, k)$ time. However, implementing the full topological sort of the exponentially many DFAs using only polynomial space requires further attention.

\subsection{A space-efficient topological sort}
\label{sec:space-efficient-topo-sort}

Recall that the standard algorithm for topological sorting performs a depth-first search over a directed graph whose nodes consists of the DFAs in the collection $\mathcal{C}_{s,k}$. In the worst case, this approach requires maintaining a history of exponentially many already-enumerated DFAs, which clashes with our memory constraints. To remedy this, we consider a \textit{Rank Iteration Strategy} for enumerating $\mathcal{C}_{s,k}$ in a space-efficient manner.

\SetAlgoNoEnd
\begin{algorithm}[t]
\caption{Rank Iteration Strategy (recursive doubling for rank)}
\label{algo:rank-iteration-strategy}

\newcommand{\EnumAll}{Let $B_1,\dots,B_N$ be any enumeration of all DFAs of size at most $s$}

\noindent
\begin{minipage}[t]{0.49\linewidth}
\small
\SetKwProg{Proc}{Procedure}{:}{}
\SetKwFunction{rankIteration}{rankIteration}

\KwIn{DFA size $s$, strict partial order $\prec$,\\ 
\hspace*{\algorithmicindent}\hspace*{\algorithmicindent}position $j\le N$}

\KwOut{DFA $A_j$ that is $j^\text{th}$ in a topological sorting of all DFAs of size $\le s$ under $\prec$}
\vspace{0.2em}
\Proc{\rankIteration{$s,\prec,j$}}{
  Initialize $\mathrm{count}\gets 0$\;
  
  \For{$r=0,\dots,N-1$}{
    \EnumAll\;
    
    \For{$i=1,\dots,N$}{
      \If{\textup{\texttt{computeRank}}$(B_i,\prec)=r$}{
        $\mathrm{count}\gets \mathrm{count}+1$\\
        \lIf{$\mathrm{count}=j$}{\Return $B_i$}
      }
    }
  }
}
\end{minipage}\hfill
\begin{minipage}[t]{0.49\linewidth}
\small
\SetKwProg{Proc}{Procedure}{:}{}
\SetKwFunction{computeRank}{computeRank}

\KwIn{DFA $D$, strict partial order $\prec$}
\KwOut{$\rank(D,\prec)$ as in \eqref{eqn:rank-def}}
\vspace{0.2em}
\Proc{\computeRank{$D,\prec$}}{
  \For{$\ell=N-1,N-2,\dots,1$}{
    \EnumAll\;
    
    \For{$i=1,\dots,N$}{
      \lIf{\textup{\texttt{path}}$(B_i,D,\ell,\prec)$}{\Return $\ell$}
    }
  }
  \Return $0$\;
}
\end{minipage}

\vspace{1em}

\noindent
\begin{minipage}[t]{\linewidth}
\small
\SetKwProg{Proc}{Procedure}{:}{}
\SetKwFunction{path}{path}

\KwIn{DFAs $C,D$, path length $\ell\ge 1$, strict partial order $\prec$}
\KwOut{True iff there exists a $\prec$-chain of length $\ell$ from $C$ to $D$}
\vspace{0.2em}
\Proc{\path{$C,D,\ell,\prec$}}{
  \lIf{$\ell=1$ \textbf{and} $C\prec D$}{\Return True}
  \EnumAll;\; let $m\gets \lceil \ell/2\rceil$\;
  
  \For{$i=1,\dots,N$}{
    \If{\textup{\texttt{path}}$(C,B_i,m,\prec)$ \textbf{and}
        \textup{\texttt{path}}$(B_i,D,\ell-m,\prec)$}{
      \Return True
      \vspace{0.3em}
    }
  }
  \Return False\;
}
\end{minipage}
\end{algorithm}

Concretely, given an arbitrary strict partial order $\prec$, for any DFA $A$ in the collection $\mathcal{C}_{s,k}=\{A_1,\dots,A_N\}$, define its \textit{rank} with respect to $\prec$, denoted $\rank(A, \prec)$, as follows:
\begin{equation}
    \label{eqn:rank-def}
    \rank(A, \prec) := \max\{\ell : \exists A_{i_1},\dots,A_{i_\ell} \in \mathcal{C}_{s,k} \text{ such that } A_{i_1} \prec \dots \prec A_{i_\ell} \prec A\}.
\end{equation}
Here, $\rank(A, \prec)=0$ if there does not exist
$A_{i_1} \in \mathcal{C}_{s,k}$ satisfying $A_{i_1} \prec A$. The rank of any $A \in \mathcal{C}_{s,k}$ can range from $0$ to $N-1$. The following elementary observation relates rank to the partial order $\prec$:
\begin{observation}[Relating $\rank$ to $\prec$]
    \label{obs:rank-to-po}
    If $A' \prec A$, then $\rank(A, \prec) \ge \rank(A', \prec)+1$.
\end{observation}
This observation implies that, in order to traverse the DFAs topologically sorted according to $\prec$, it suffices to iterate through them in increasing order of their rank. This gives rise to the Rank Iteration Strategy encapsulated in Algorithm \ref{algo:rank-iteration-strategy}. The strategy processes DFAs in increasing order of their ranks. The rank of each DFA is computed space-efficiently using a technique inspired by the ``middle-first search'' approach used in Savitch's theorem \citep{savitch1970relationships}. Savitch's theorem gives a space-efficient algorithm for checking if two nodes are connected by a path of length $\ell$ in a directed graph, by recursively checking for the existence of a ``middle node'' that lies at distance $\ell/2$ from each of the nodes. The crucial point is that these recursive checks can reuse space. A similar idea applies for our purpose of computing rank, where a path of length $\ell$ between two DFAs $C$ and $D$ corresponds to a chain $C \prec B_1 \prec \dots \prec B_{\ell-1} \prec D$ in the partial order. In this case, if the relation $\prec$ can be checked using polynomial space, then the entire Rank Iteration Strategy uses only polynomial space. This leads to the following proposition, whose proof is given in \Cref{sec:appendix-upper-bounds}.
\begin{restatable}[Rank Iteration Space Complexity]{proposition}{RankIterationSpaceComplexity}
    \label{prop:rank-iteration-space-cxty} \textup{\texttt{rankIteration}}$(s,\prec,j)$ given in \Cref{algo:rank-iteration-strategy} requires $\p(s, k)$ bits of memory, provided one can check ``$C \prec D$?'' using $\p(s, k)$ space for any two DFAs $C$ and $D$ with at most $s$ states.
\end{restatable}

\subsection{A natural space-efficient modification and the main challenge}
\label{sec:modification-and-challenge}

We now tackle the issue of storing input history. The subroutine \mbox{\effiter$(\prec)$} below modifies Angluin’s algorithm to avoid storing the full input history $\{w_1,\dots,w_t\}$ up to time $t$. However, this modification will necessitate changing the partial order $\prec_1$ considered by Angluin's algorithm, as the proceeding analysis reveals. 

\begin{framed}
    \noindent\effiter$(\prec)$: Let $A_1,\dots,A_N$ denote all the DFAs of size at most $s$ topologically sorted according to $\prec$. The previous analysis shows that we can traverse the DFAs in this order space-efficiently.
    Let $i(0)=1$. At any time step $t \ge 1$, the algorithm first initializes $i(t)=i(t-1)$. After receiving the streamed input string $w_t$, the algorithm checks if $w_t \in L(A_{i(t)})$, which only requires tracing $w_t$ through the DFA $A_{i(t)}$, and hence requires constant memory. Here, $A_{i(t)}$ is retrieved by invoking \textup{\texttt{rankIteration}}$(s,\prec,i(t))$. If this check does not pass, the algorithm increments $i(t) \gets i(t)+1$ and outputs the DFA $A_{i(t)}$. Note that the outputted $A_{i(t)}$ might not accept $w_t$. Indeed, requiring the algorithm to retest membership of $w_t$ in $L(A_{i(t)})$ could necessitate remembering a prohibitively long string $w_t$.
\end{framed}
We can instantiate \effiter\ with the partial order $\prec_1$, and by the preceding discussion, this procedure uses only $\p(s,k)$ memory. The remaining question is whether it identifies $K$ in the limit; i.e., whether it eventually outputs $A_{i(t)}$ with $L(A_{i(t)})=K$ for all $t \geq t^\star$.

Recall that there exists an index $i^\star$ such that $L(A_{i^\star})=K$. Now, if the algorithm ever arrives at index $i^\star$, it never proceeds beyond it, since every $w_t \in K$. Thus, in the limit, the algorithm converges to some index $i(t) \le i^\star$. Then, it must be the case that $|K \setminus L(A_{i(t)})| < \infty$; otherwise, we are guaranteed to see some string $w \in K \setminus L(A_{i(t)})$ in the future, contradicting convergence at $i(t)$.

If $i(t)=i^\star$, then we identify $K$. Otherwise, $i(t) < i^\star$. We already argued above that $|K \setminus L(A_{i(t)})| < \infty$. If $|K \setminus L(A_{i(t)})| = 0$, then  $K \subseteq L(A_{i(t)})$. In that case, we are guaranteed that $K = L(A_{i(t)})$. Indeed, $K \subset L(A_{i(t)})$ would contradict the definition of $\prec_1$. So, we achieve the desired objective in this case as well.

The only case that remains is when $i(t) < i^\star$ and $|K \setminus L(A_{i(t)})| < \infty \neq 0$. Unfortunately, in this case, it is possible that $L(A_{i(t)}) \setminus K$ is non-empty as well. In fact, $L(A_{i(t)})$ could even contain infinitely many strings outside $K$. So, in this case, not only do we fail to achieve the desired objective of identification, but the language that we converge to may contain infinitely many hallucinations.

This unwanted case is a direct consequence of the memory constraint. Indeed, when $K \setminus L(A_{i(t)})$ is nonempty, and if we were to store the entire history of inputs at every time step (as Angluin's algorithm does), the history would eventually contain some string in $K \setminus L(A_{i(t)})$, causing us to proceed beyond $A_{i(t)}$. However, consider the possibility where $K \setminus L(A_{i(t)})$ has very few strings which all appeared in the distant past without showing up again. Since we do not store past inputs in memory, and only make our decisions based on the most recent string, we never see any evidence again that would cause us to move beyond $A_{i(t)}$.

As a final remark, note that \effiter$(\prec_1)$ \emph{does} succeed with exact identification under the additional assumption that every string in $K$ appears infinitely often in the input stream. In this setting, any missing string $w\in K\setminus L(A_{i(t)})$ would eventually reappear and invalidate the currently wrong hypothesis.

\subsection{A slightly more robust partial order}
\label{sec:robust-po}
The analysis above isolates the main limitation of the partial order $\prec_1$: If there exist two DFAs $A$ and $A'$ such that $|L(A) \setminus L(A')| < \infty$ but $|L(A') \setminus L(A)|$ is huge, then topologically sorting according to $\prec_1$ allows $A'$ to be placed before $A$. Consequently, if $L(A)$ is the target language, we might converge at $L(A')$, which contains a large number of hallucinations outside $L(A)$. Our final solution arises as a direct remedy to this issue. 

Concretely, consider an ordering of the DFAs which enforces the following property: For any two DFAs $A$ and $A'$, if $|L(A) \setminus L(A')| < \infty$ and $|L(A) \setminus L(A')| < |L(A') \setminus L(A)|$, then $A$ appears before $A'$. Namely, we would like to consider the relation
\begin{equation}
    \label{eqn:po2-def}
    A \prec_2 A' \iff |L(A) \setminus L(A')| < \infty \text{ and } |L(A) \setminus L(A')| < |L(A') \setminus L(A)|.
\end{equation}
If this relation were a strict partial order, we could topologically sort the DFAs according to $\prec_2$ and run the space-efficient algorithm described above on this ordering. As it turns out, the relation is indeed a strict partial order; the elementary proof appears in \Cref{sec:appendix-upper-bounds}.
\begin{restatable}[$\prec_2$ Strict Partial Order]{proposition}{popartialorder}
    \label{prop:po2-strict-partial-order}
    The relation $\prec_2$ defined in \eqref{eqn:po2-def} is a strict partial order.
\end{restatable}
Next, we observe that checking $A \prec_2 A'$ can also be done with $\p(s,k)$ space. For this, we once again construct the product automata that recognize $L(A) \setminus L(A')$ and $L(A') \setminus L(A)$, respectively (each of these have at most $s^2$ states). Then, we use the fact that a regular language is infinite if and only if an automaton recognizing it has a cycle that can be reached from the initial state, and can then go on to reach an accepting state. We can check this using standard depth-first search on the directed graph underlying the automaton. For our purposes, this requires $\p(s,k)$ time, and hence space. Finally, if both the languages $L(A) \setminus L(A')$ and $L(A') \setminus L(A)$ are determined to be finite, we require comparing their sizes. In this case, we observe that both languages must contain strings of length at most $s^2$. Otherwise, if either contained a string of length larger than $s^2$, this string would induce a cycle while traversing its product automaton, which could be pumped infinitely. So, the size of both languages can be computed exactly by iterating over all strings of length at most $s^2$, and checking membership in the automaton. This step can also be executed in $\p(s,k)$ space.

Using all the ingredients established so far, we can then show that \effiter, invoked with the partial order $\prec_2$, converges to a DFA that has a finite symmetric difference with the target language $K$ --- this notion has been referred to as \textit{hyper-equivalence} in the automata literature~\citep{badr2009hyper_a, gawrychowski2011minimising, maletti2011optimal}.
\begin{proposition}[$\prec_2$ Converges to Finite $\triangle$]
    \label{prop:po2-converges-to-good-language}
    \effiter$(\prec_2)$ converges to an index $i(t)$ satisfying $|L(A_{i(t)}) \triangle K| < \infty$, while using only $\p(s,k)$ memory. Furthermore, if every string in $K$ shows up infinitely often in the input stream, then $L(A_{i(t)})=K$ in the limit. 
\end{proposition}
\begin{proof}
    Since it is guaranteed that $K=L(A_{i^\star})$ for some index $i^\star$, the algorithm never proceeds beyond $i^\star$, and hence, $i(t) \le i^\star$ in the limit. Furthermore, $|K \setminus L(A_{i(t)})| < \infty$, since otherwise, the algorithm eventually sees a string that causes it to move past $i(t)$. So, either $i(t)=i^\star$, or $i(t) < i^\star$. In the latter case, it must hold that $|L(A_{i(t)})\setminus K| \le |K \setminus L(A_{i(t)})| < \infty$, since we topologically sorted according to $\prec_2$. Either way, $|L(A_{i(t)}) \triangle K| < \infty$ holds. Since we argued above that checking $\prec_2$ uses $\p(s,k)$ space, by Proposition \ref{prop:rank-iteration-space-cxty}, the entire algorithm uses only $\p(s,k)$ space.

    Finally, if every string in $K$ shows up in the input infinitely often, then the algorithm converges to an $i(t)$ satisfying $L(A_{i(t)}) \supseteq L(A_{i^\star})$. Otherwise, a string in $L(A_{i^\star}) \setminus L(A_{i(t)})$ is guaranteed to show up in the input, causing the algorithm to move beyond $i(t)$. Now, suppose that $L(A_{i(t)}) \supset L(A_{i^\star})$. Then, we have that $|L(A_{i^\star}) \setminus L(A_{i(t)})|=0<|L(A_{i(t)}) \setminus L(A_{i^\star})|$, which contradicts sorting according to $\prec_2$ since $i(t) < i^\star$. Thus, $L(A_{i(t)})=L(A_{i^\star})=K$. 
\end{proof}

\subsection{The final algorithm}
\label{sec:final-algorithm}

We now have all the ingredients necessary to state our final algorithm. By Proposition \ref{prop:po2-converges-to-good-language}, we know that \effiter$(\prec_2)$ converges to a DFA $A_{i(t)}$ that satisfies $|L({A_{i(t)}}) \triangle K| < \infty$. If we can now output a DFA that accepts a large subset of $L({A_{i(t)}})$ but avoids the strings in the finite symmetric difference $L({A_{i(t)}}) \triangle K$, our objective of language generation with finite gap would be accomplished. 

Towards this, observe that $L({A_{i(t)}}) \triangle K$ is yet another regular language. Moreover, the product DFA $C$ recognizing $L({A_{i(t)}}) \triangle K$ has at most $s^2$ states. Hence, given that this language is finite, it must not include any string of length $\ge s^2$; otherwise, this string would visit some state in the DFA $C$ twice before reaching an accepting state, yielding a cycle, which can be pumped arbitrarily often to generate infinitely many accepted strings. In fact, we can prove a stronger bound, using a prior result by \cite{gawrychowski2011minimising}, which shows that $L({A_{i(t)}}) \triangle K$ can only contain strings of length at most $2s-2$. The proof of this result appears in \Cref{sec:appendix-upper-bounds}.
\begin{restatable}[Symmetric Difference Contains Length-$O(s)$ Strings]{lemma}{linearbound}
    \label{lemma:linear-bound-symmetric-diff}
    Let $A=(Q_A, \Sigma, \delta_A, q_{0}^A, F_A)$ and $B=(Q_B, \Sigma, \delta_B, q_{0}^B, F_B)$ denote DFAs having at most $s$ states over a common alphabet $\Sigma$, such that $|L(A) \triangle L(B)| < \infty$. Then, for any string $w \in \Sigma^*$, if $|w| \ge 2s-1$, then $w \notin L(A) \triangle L(B)$.
\end{restatable}
Thus, if we output a DFA $A'$ that only accepts strings of length at least $2s-1$ in $L(A_{i(t)})$, we ensure that $L(A') \subseteq L(A_{i(t)}) \cap K \subseteq K$ as required. This leads to the following theorem, whose precise proof is given in \Cref{sec:appendix-upper-bounds}.
\begin{restatable}[Space-Efficient Language Generation]{theorem}{SpaceEfficientGenerationTheorem}
    \label{thm:space-efficient-generator}
    Let $K$ denote the unknown target language recognized by a DFA of size at most $s$. There exists an algorithm $\A$ which, for all sufficiently large $t$, outputs a DFA $A'=\A(t)$ that satisfies $L(A') \subseteq K$, and uses only $\p(s,k)$ space. Furthermore, the algorithm misses at most $|K \setminus L(A')| \le O(k^{2s-2})$ target strings.
\end{restatable}

Let us also remark here that by the conclusion of Lemma \ref{lemma:linear-bound-symmetric-diff}, our learning algorithm $\A$ is guaranteed to output an automaton $A'$ such that every string in $K \setminus L(A')$ has length at most $2s-2$. Thus, if the algorithm additionally keeps track of all the inputs of length at most $2s-2$ (which requires $\exp(s,k)$ space), $\A$ can output a DFA that exactly recognizes $K$. In other words, with $\exp(s,k)$ memory, we can identify $\mathcal{C}_{s,k}$ in the limit.

Finally, observe that our algorithm above ensures no hallucinations, but misses out on a finite amount of breadth (namely $O(k^{2s-2})$ strings) in the limit. Instead, we could have the algorithm output a different product DFA $A'$ (again on $\le 2s^2$ states), which accepts $L(A_{i(t)}) \cup \{w \in \Sigma^*:|w| \leq 2s-2\}$. This would ensure full breadth (i.e., $L(A') \supseteq K$), but a finite amount of hallucinations (again $O(k^{2s-2})$ strings). We thus have to necessarily incur a misspecification cost on $O(k^{2s-2})$ strings, either in hallucinations, or breadth. Our next result shows that this compromise is indeed unavoidable.

%% file: Sections/Lower_Bounds.tex
\section{Lower Bounds} \label{sec:lower-bounds}

In the previous section, we showed that polynomial space suffices to generate in the limit from $\mathcal{C}_{s,k}$ with a generation gap of $O(k^{2s-2})$. A natural question is whether this loss is merely an artifact of our algorithm, or whether it is fundamentally unavoidable under space constraints. In this section, we prove that the latter is actually the case.

Our result establishes that any algorithm that operates with sub-exponential memory must incur an exponential generation gap. Moreover, the same lower bound holds under the relaxed symmetric-difference objective from Definition~\ref{def:space-efficient-generation-in-the-limit}, where hallucinations and missed strings are counted together through $|L(A(t)) \triangle K|$. Thus, hallucinations cannot be traded for substantially smaller missed breadth: Any sub-exponential-memory learner still incurs an exponential total error. This demonstrates a sharp transition between what is achievable with polynomial space versus exponential space: Exponential memory permits full identification, whereas any polynomial-space learner is fundamentally forced to miss an exponential number of strings.

Our proof departs from the diagonalization-based arguments commonly used in the identification-in-the-limit literature. Instead, we adopt a communication-complexity perspective and reduce from the $\textnormal{INDEX}_n$ problem. At a high level, we show that a space-efficient learner with a sufficiently small generation gap could be used to derive an efficient communication protocol for the $\textnormal{INDEX}_n$ problem, which contradicts known lower bounds. We begin by recalling the relevant communication problem and its standard lower bound, and then we present the reduction that translates a small generation gap into an efficient Index protocol.
\begin{lemma}[Communication Lower Bound for $\textnormal{INDEX}_n$, \cite{kushilevitz1997communication}]
\label{lem:index-lb}
    Let $\textnormal{INDEX}_n$ be the one-way communication problem where Alice is given a vector
    $x \in [k]^n$, Bob is given an index $i \in [n]$, and Bob must output $x_i$ after a single message from Alice. Any deterministic protocol for $\textnormal{INDEX}_n$ requires
    at least $n\log_2k$ bits of communication.
\end{lemma}

The proof of Lemma \ref{lem:index-lb} is a standard pigeonhole argument; for completeness we include it in \Cref{sec:appendix-lower-bounds}. The next lemma connects $\textnormal{INDEX}_n$ to generation in the limit over $\mathcal{C}{s,k}$, yielding a generation-gap lower bound in terms of the learner’s space.

\begin{lemma}\label{lem:space-lower-bound-analytical}
    Any algorithm $\mathcal{A}$ that uses $m(s)$ bits of space and generates in the limit from $\mathcal{C}_{s,k}$ must incur a generation gap of
    $
    \Delta_\mathcal{A}(s,k) \geq k^{s - 2\log_k((m(s)+1)/\log_2k)-3}.
    $
\end{lemma}
\newcommand{\map}{\textnormal{map}}
\newcommand{\Base}{\textnormal{Base}}
\begin{proof}
Fix $s$ and let $m(s)$ be the number of bits that algorithm $\mathcal{A}$ uses.
As we discussed above our proof goes through the $\textnormal{INDEX}_n$ problem. We will 
set $n$, i.e. the input size of $\textnormal{INDEX}_n$, later in the proof as a function of
$m(s)$ and $k$.

\paragraph{Alphabet and building blocks.} Let $\Sigma \coloneqq [k]$ be an alphabet of size $k$. 
Throughout the proof we use $[k] = \{0,1,\dots,k-1\}$. We define $U$ to be the set of all strings that start with a $1$ and continue with some finite number of symbols from $\Sigma$.
We use the first, say $p$, symbols after $1$ to encode a pair $(i,q)$ for all $i \in [n]$
and all $q \in \Sigma$, and we  partition the universe $U$ into subsets based on these $p$ symbols. Then, every string will continue with some free symbols, say $\ell$ of them, which means that
each of the subsets corresponding to a pair $(i,q)$ has $k^\ell$ strings. 
Since there are $kn$ different pairs $(i,q)$ and we are encoding in base-$k$, we need $p \coloneqq 1 + \log_k n$ symbols to encode them.
Formally, for each
$j \in \{0,1,\dots, kn-1\}$, let $\textnormal{enc}_k(j) \in \Sigma^p$ denote the base-$k$ encoding of $j$ with length $p$ (allowing leading zeros). We define the languages $L_0, \dots, L_{kn-1}$, which partition $U$, as
$$
L_j \coloneqq \{ 1\ \circ\ \textnormal{enc}_k(j)\ \circ\ u\ |\ u \in \Sigma^\ell\}.
$$
Every language above has size $k^\ell$, as they have exactly one string for every member of $\Sigma^\ell$.
We create a bijective mapping, $\map: [n] \times [k] \to [kn]$, that maps every combination of a dimension $i \in [n]$ and symbol $q \in [k]$ to a unique language in the above partition. Any such mapping can work for our reduction, for simplicity let $\text{map}(i,q) \coloneqq i\cdot k+ q$.

Furthermore, for every $i \in [n]$ and every $q \in [k]$ we define the language $\Base(i,q)$ to contain all above partitions of $U$ except for $L_{\map(i,q)}$. For technical reasons that will become clear later, we also make $\Base(i,q)$ include an infinite component, namely all strings that start with a $0$. Formally,
$$
\Base(i,q) \coloneqq \{0 \circ u\ |\ u \in \Sigma^* \}\ \cup\ \{1 \circ u\ |\ u \in \Sigma^{p+\ell}\} \setminus L_{\textnormal{\map}(i,q)}.
$$
Notice that all partitions of $U$, i.e. $L_0, \dots, L_{kn-1}$, can be recognized by a DFA with at most $p+\ell+2$ states, which simply recognizes the encoding bits and then counts $\ell$ more characters. Also, for each language $\Base(i,q)$, we can construct a DFA that recognizes it as follows:
\begin{itemize}[itemsep=2pt,topsep=2pt,parsep=0pt,partopsep=0pt]
    \item If the first symbol is $0$, transition to an accepting state that self-loops on all symbols in $\Sigma$ (this branch accepts $\{0 \circ u\ |\ u \in \Sigma^*\})$.
    \item If the first symbol is $1$, the DFA needs to accept exactly $U \setminus L_{\map(i,q)}$.
    Equivalently, among strings of length $1+p+\ell$ that start with $1$, it rejects exactly those whose
    next $p$ symbols encode the number $\map(i,q)$, i.e. $\textnormal{enc}_k(\map(i,q))$. This can
    easily be done by $2$ chains of length $p$ that track whether or not the sequence encodes the target number. After this part, the chain that recognized this sequence leads to a rejecting state, while the chain that did not recognize the sequence reads $\ell$ more symbols and accepts.
\end{itemize}

\noindent In Appendix~\ref{sec:appendix-lower-bounds}, we give a figure (\Cref{fig:dfa_accepting_base}) that shows the above construction. Such a DFA uses at most $2p+\ell+5$ states. We now set $s\coloneqq 2p+\ell+5$ so that all above languages, i.e.
all $L_{\map(i,q)}$ and $\Base(i,q)$ for $i \in [n]$ and $q \in \Sigma$, are recognized by DFAs of size at most $s$.

\paragraph{Reduction from $\textnormal{INDEX}_n$.}
We now build a deterministic one-way protocol for $\textnormal{INDEX}_n$ using algorithm $\mathcal{A}$. As a reminder, Alice receives $x = (x_0, x_1, \dots, x_{n-1}) \in \Sigma^n$, while Bob receives $i \in [n]$ and must output $x_i$. Alice simulates $\mathcal{A}$ on a finite input stream consisting of all the strings
in the languages $L_{\map(j,q)}$ for every coordinate $j \in [n]$ and every symbol $q \neq x_j$. Conceptually, the only partitions of $U$ that are \textbf{not} fed to the algorithm are exactly the combinations $(j,x_j)$ of Alice's input. Formally, she feeds $\mathcal{A}$ the set
$$
S_A = \bigcup_{j \in [n]} \bigcup_{q \in \Sigma \setminus \{x_j\}} L_{\map(j,q)}.
$$
Let $\sigma$ be the internal memory state of $\mathcal{A}$ after reading this prefix.
Alice sends the state $\sigma$ to Bob, which by definition is at most $m(s)$ bits. Bob resumes the execution of $\mathcal{A}$ from state $\sigma$ and continues by feeding it all the strings in the languages
$L_{\map(j,q)}$ for $j \neq i$, i.e., the set

$$
S_B = \bigcup_{\substack{j \in [n], \\ j \neq i}} \bigcup_{q \in \Sigma} L_{\map(j,q)}.
$$

\noindent Notice that at this point, the learner $\mathcal{A}$ has seen all partitions of $U$ except for the one corresponding to $(i,x_i)$. After feeding these strings, Bob continues running the algorithm $\mathcal{A}$ and gives it strings from $\{0 \circ u\ |\ u \in \Sigma^*\}$, via some surjective enumeration. This part is simply done so that the
algorithm $\mathcal{A}$ can keep running until it converges to an answer.

Consider now the union of the strings presented to $\mathcal{A}$, i.e. $S_A \cup S_B \cup \{0 \circ u\ |\ u \in \Sigma^*\}$. For all $j \neq i$, Bob has given $\mathcal{A}$ all the strings in the languages $L_{\map(j,q)}$. Furthermore, for the index $j = i$, Alice has given $\mathcal{A}$
all the strings in the languages $L_{\map(i,q)}$ for all $q \neq x_i$. Thus, the algorithm has seen all languages $L_{\map(j,q)}$ except for $L_{\map(i,x_i)}$. Formally, the stream is a surjective
enumeration of
$$
S_A \cup S_B \cup \{0 \circ u\ |\ u \in \Sigma^*\} = \{0 \circ u\ |\ u \in \Sigma^*\}\ \cup\ \{1 \circ u\ |\ u \in \Sigma^{p+\ell}\} \setminus L_{\textnormal{\map}(i,x_i)} = \Base(i,x_i).
$$
\textbf{Decoding $x_i$ from algorithm $\mathcal{A}$.} By the definition of generation in the limit,
with gap $\Delta_{\mathcal{A}}(s,k)$, applied to $L^* = \Base(i,x_i)$, the limiting output of $\mathcal{A}$, let it be $H$, satisfies
\[
L(H)\subseteq \mathrm{Base}(i,x_i)
\qquad\text{and}\qquad
\bigl|\mathrm{Base}(i,x_i)\setminus L(H)\bigr|\le \Delta_\mathcal{A}(s,k).
\]
Let's assume at this point that the generation gap of the learner $\mathcal{A}$ is 
$\Delta_{\mathcal{A}}(s,k) < k^\ell$.
Then, for any $q \neq x_i$, if $L(H) \cap L_{\map(i,q)} = \emptyset$, then $\mathcal{A}$ would miss all $k^\ell$ strings in $L_{\map(i,q)} \subseteq \Base(i,x_i)$, which would mean that its generation gap
is at least $|\Base(i,x_i) \setminus L(H)| \geq k^\ell$, contradicting our assumption. Therefore, for every $q \neq x_i$ we must have
$
L(H) \cap L_{\map(i,q)} \neq \emptyset.
$
On the other hand, $L_{\map(i,x_i)}$ is disjoint from $\Base(i,x_i)$, and since
$L(H) \subseteq \Base(i,x_i)$, it holds that
$
L(H) \cap L_{\map(i,x_i)} = \emptyset.
$

As a result, $x_i$ is the unique symbol $q \in \Sigma$ such that $L(H) \cap L_{\map(i,q)} = \emptyset$.
Bob can, thus, determine $x_i$ by testing the emptiness of the intersection of $L(H)$ with the languages $L_{\map(i,q)}$ (e.g. by constructing the product DFAs and checking if there is any path from the starting state to an accepting state). Notice that Bob does not necessarily need to know when the learner $\mathcal{A}$ has converged. The reduction from $\textnormal{INDEX}_n$ is purely information-theoretic and therefore it suffices that
there exists a $t^*$ after which Bob could test the above conditions and recover $x_i$. See~\Cref{sec:appendix-lower-bounds} for further discussion.

We have, therefore, constructed a deterministic one-way protocol for $\textnormal{INDEX}_n$
in which Alice's message $\sigma$ has at most $m(s)$ bits. Choosing $m(s) = n\log_2k-1$
leads to a contradiction of the $\textnormal{INDEX}_n$ lower bound (\Cref{lem:index-lb}),
which means that our assumption of $\Delta_{\mathcal{A}}(s,k) < k^\ell$ is false
in the case where $m(s) = n\log_2k -1$.
Recall that we used $s = 2p+\ell+5$ and $p = 1 + \log_kn$, which yields
$\ell = s - 2p - 5$ and $p = 1 + \log_k((m(s)+1)/\log_2k)$. The proof of the lemma is concluded by combining the former identities
$
\Delta_\mathcal{A}(s,k) \geq k^\ell = k^{s - 2\log_k((m(s)+1)/\log_2k)-3}.
$
\end{proof}
Finally, the main result of this section, \Cref{thm:intro-main-lower-bound}, follows
from \Cref{lem:space-lower-bound-analytical}, as in the case where $\Delta_{\mathcal{A}}(s,k) \leq k^{(1-\epsilon)s}$, we have that $2\log_k((m(s)+1)/\log_2k) \geq \varepsilon s - 3$, i.e. $m(s) = k^{\varepsilon s/2 - 3 + \log_k(\log_2 k)} = k^{\Omega(\varepsilon s)}$.

%% file: Sections/Conclusion.tex
\section{Conclusion}
This paper develops a space-efficient model of generation in the limit for the DFA family $\mathcal{C}_{s,k}$ and proves a tight memory–breadth tradeoff: $\mathrm{poly}(s,k)$ space guarantees hallucination-free generation up to an exponential gap, while achieving substantially smaller gaps (and, in particular, exact identification) forces exponential memory.

Several directions remain open, including extending the framework to non-uniform learners, and developing analogous tradeoffs for other resources such as time and sample complexities.
\label{sec:conclusion}

%% file: Appendix_Sections/Communicative_Languages.tex
\section{Resource-Bound Communication}
\label{app:communicative}

Here, we expand on our point from \Cref{sec:introduction} that space-efficient learning in the limit admits a meaningful formulation only over finite hypothesis spaces of regular languages. We also argue that this language-acquisition regime captures the most relevant and informative setting.

\paragraph{Space-Efficient Generation.} Under our proposed framework, the learner $\A$ receives an adversarial enumeration $w_1, w_2, w_3, \dots$ of a target language $K$ through an online streaming interface. Each string $w_t \in K$ arrives one symbol at a time, with a designated delimiter separating consecutive strings. After each read symbol, the learning algorithm can perform computation and storage within a specified memory budget, dependent on the complexity of the target language. Furthermore, after processing each $w_t$, the learner outputs a hypothesis representation. The generation objective requires that, after some finite time, every subsequent output of $\A$ must represent a language $L \subseteq K$.

As a consequence of the Church-Turing Thesis \citep{sep-church-turing}, we can model space-bounded learners as \textit{streaming Turing machines} (STM) \citep{aggarwal2007data} with bounded-size work tapes. An STM features a one-way read-only input and output tapes and several bidirectional work tapes. Clearly, a space-bounded STM admits only finitely many reachable configurations, so any decision procedure implemented by such a model necessarily recognizes a regular language. Now, the space-bounded learner $\A$ who searches for the target $K$ within a hypothesis collection $\C_\A$ should have the ability to verify whether a given input stream $w \in K$ belongs to the current hypothesis $L \in \C_\A$. Hence, since the finite-memory STM $\A$ decides membership queries for every hypothesis language, $\C_\A$ must contain only regular languages. Moreover, since $\A$ can occupy finitely many internal configurations, $\A$ can only decide membership for finitely many languages, leading to $|C_\A| < \infty$. Thus, space-efficient generation in the limit naturally confines analysis to finite collections of regular hypotheses.  

Although the restriction to finite regular hypothesis spaces may appear limiting, this language-acquisition regime provides the central case for understanding realistic learners. In particular, we demonstrate that the collection $\mathcal{K}$ of \textit{communicative languages}, which arise from the interaction of space-bounded agents such as humans \citep{miller1963finitary}, forms a strict subset of regular languages.

\paragraph{Regularity of Communicative Languages.}
We now show from the dual perspectives of space-bounded language users and learners that the set of communicative languages $\mathcal{K}$ can only contain regular languages.

First, we posit that a competent user $U$ of a language $K \in \mathcal{K}$ has the ability to decide whether $w \in K$ for any input stream $w \in \SS$. Together with the finite memory assumption for $U$, this assertion implies the regularity of $K$. A potential objection arises from Chomsky’s classic observation that the possibility of unbounded center embedding makes English at least context-free \citep{chomsky1956three}. However, in practice, an English speaker will need access to a pen and paper to reliably judge the grammaticality of deeply center-embedded sentences such as
“The mouse that the cat that the dog chased bit ran.” Indeed, for any input stream $w \in \SS$, humans respond in one of three ways: (i) accept $w$ as grammatical, (ii) reject $w$ as ungrammatical, or (iii) determine that processing $w$ requires additional resources and reject $w$ as an unnatural conversational sentence. In this sense, the set of grammatically well-formed English strings $\bar{K}$ might sit high in the Chomsky hierarchy, but the communicative English sub-language $K \subset \bar{K}$ remains regular.

Second, we posit that any communicative language $K \in \mathcal{K}$ admits learning by a space-bounded learner $\A$. Indeed, human infants acquire grammar without external memory aids, just as large transformers successfully train with finite context windows. Hence, as argued above, the hypothesis space $\C_\A$ contains finitely many regular languages -- a conclusion corroborated by the search spaces of LLMs.\footnote{Transformers, RNNs, and related NTP architectures behave like large probabilistic deterministic finite-state automata \citep{svete2023recurrent, svete2024transformers, merrill2024illusion} and therefore generate from a regular support.} Therefore, due to the realizability assumption $K \in \C_\A$, $K$ becomes regular.

Importantly, our arguments for regularity do not imply that communicative languages have finite size or strings of bounded length. Throughout, we model strings as concatenations of grammatically well-formed sentences similar to token sequences produced by LLMs, which may consist of many sentences terminated by an end-of-sequence marker. Under this interpretation, a communicative language may contain strings of unbounded length, provided that those strings remain locally parseable by a finite-memory process. As an analogy, a proficient human reader can judge an entire book as grammatically correct while using only bounded working memory. Finally, we emphasize that we treat languages purely as sets of grammatically well-formed strings and deliberately ignore semantic verification, which may require substantially greater resources.

%% file: Appendix_Sections/DFA_Stuff.tex
\section{Standard DFA Results}
\label{sec:appendix-dfa-facts}

In this section, we briefly state some foundational and some more recent results about DFAs. For a more detailed presentation, we refer the reader to the standard books on Automata Theory \citep{sipser1996introduction, hopcroft2001introduction, esparza2023automata}.

For a DFA $A = (Q, \Sigma, \d, q_0, F)$, we denote by $L(A) = \{w \in \SS : \d(q_0, w) \in F\}$ the regular language accepted by $A$ and let
\begin{equation*}
    L_A^-(q) = \{w \in \SS : \d(q_0, w) = q \} \quad \text{ and } \quad L_A^+(q) = \{w \ \in \SS : \d(q, w) \in F\}
\end{equation*}
stand for the left (or prefix) and right (or residual) languages of state $q \in Q$. Throughout the paper, we use the term \textit{automaton} to refer to a DFA.

Regular languages possess an interesting algebraic property described by the Myhill-Nerode theorem \citep[Chapter 2.4]{esparza2023automata}. For a subset $L \subseteq \SS$, we can define the Nerode equivalence $\equiv_L$ between strings in $\SS$ such that  $u \equiv_L v \iff \{w \in \SS: u \cdot w \in L\} =  \{w \in \SS: u \cdot w \in L\}$. Now, the theorem states that $L$ is a regular language if and only if $\equiv_L$ splits $\SS$ into a finite number of equivalence classes, indicated as $\vert \!\! \equiv_L \!\! \vert$. Moreover, when $\vert \!\! \equiv_L \!\! \vert < \infty$, there exists a unique minimum-state DFA for $L$ with the following description:
$q_0 = [\eps]_{\equiv_L}, \ Q = \{[w]_{\equiv_L} : w \in \SS\}, \ F = \{[w]_{\equiv_L} : w \in L\}, \ \d([w]{\equiv_L}, \s) = [w \cdot \s]_{\equiv_L}, \en \forall w \in \SS, \s \in \Sigma$.

One can also consider equivalence of DFAs in terms of accepting languages: $A \equiv B \iff L(A) = L(B)$. This equivalence extends to automata states. Indeed, for a DFA $A = (Q, \Sigma, \d, q_0, F)$, we define states $p, q \in Q$ as equivalent ($p \equiv q)$ if $L_A^+(p) = L_A^+(q)$. Given the automaton $A$, Hopcroft's algorithm \citep{hopcroft1971n, gries1973describing} allows us to find the minimal automaton for $L(A)$ in $O(|\Sigma||Q|\log |Q|)$ time by cleverly merging equivalent states. As a consequence of Hopcroft's algorithm, if $p \not\equiv q$, then there exists a distinguishing string $w \in L_A^+(p) \triangle L_A^+(q)$ of length at most $|Q|-1$.

A relaxed variant of equivalence, referred to as hyper-equivalence \citep{badr2009hyper_a, gawrychowski2011minimising, maletti2011optimal}, underpins our results. We define two automata $A$ and $B$ as hyper-equivalent, denoted $A \sim B$, if $L(A)$ and $L(B)$ disagree on finitely many strings: $|L(A) \triangle L(B)| < \infty$. Clearly, $\sim$ is an equivalence relation, which we can again extend to automata states by writing $p \sim q$ if and only if $|L_A^+(p) \triangle L_A^+(q)| < \infty$.

Now, we describe the standard product DFA constructions. In the notation of \Cref{def:dfa}, for two DFAs $C=(Q_1,\Sigma, \delta_1, q^{(1)}_0, F_1)$ and $D=(Q_2,\Sigma, \delta_2, q^{(2)}_0, F_2)$ accepting regular languages $L(C)$ and $L(D)$, the DFAs recognizing $L(C) \cap L(D)$, $L(C) \cup L(D)$ and $L(D) \setminus L(C)$ can be realized as a suitable product automaton $C \times D$. The state space in this automaton is $Q_1 \times Q_2$ (so that the number of states is $|Q_1|\cdot|Q_2|$), and the initial state is the pair $(q^{(1)}_0, q^{(2)}_0)$. The transition function $\delta_1 \times \delta_2$ is defined as
\begin{align*}
    (\delta_1 \times \delta_2)((q_1, q_2), \sigma)=(\delta_1(q_1,\sigma), \delta_2(q_2,\sigma)).
\end{align*}
For $L(C) \cap L(D)$, the accepting states correspond to $(q_1,q_2)$ pairs such that both $q_1 \in F_1$ and $q_2 \in F_2$. For $L(C) \cup L(D)$, the accepting states correspond to $(q_1,q_2)$ pairs such that either $q_1 \in F_1$ or $q_2 \in F_2$. For $L(D) \setminus L(C)$, the accepting states correspond to $(q_1,q_2)$ pairs such that $q_1 \notin F_1$ and $q_2 \in F_2$.

%% file: Appendix_Sections/Appendix_Upper_Bounds.tex
\section{Deferred Proofs from \Cref{sec:upper-bounds}}
\label{sec:appendix-upper-bounds}

\RankIterationSpaceComplexity*
\begin{proof}
    Within \textup{\texttt{rankIteration}}$(s,\prec,j)$, keeping track of $\mathrm{count}$, and the current value of $r$ in the outer for loop requires $\log(N)=\p(s, k)$ space. The inner loop, which traverses the collection of DFAs, also requires keeping track of only $\p(s, k)$ bits of memory. For example, this can be implemented by keeping track of $s$ bits indicating the states present in the DFA and $sk \log s$ bits indicating the transitions at each state and alphabet symbol, followed by elementary checks to ensure that the implied DFA is valid. Thereafter, observe that we can reuse the space required by \texttt{computeRank}$(\cdot, \prec)$ in every iteration of the nested for loop. Thus, the total space complexity of Algorithm \ref{algo:rank-iteration-strategy} is $\p(s, k)$ \textit{plus} the space required to implement \texttt{computeRank}$(\cdot, \prec)$.

    We now argue that \texttt{computeRank}$(\cdot, \prec)$ can be implemented efficiently. Again, the outer for loop requires keeping track of $\ell$ which requires $\log(N)=\p(s,k)$ space, and the inner for loop which traverses the collection of DFAs also requires keeping track of $\p(s,k)$ bits, as argued above. Thereafter, every iteration of the nested for loop can reuse the space required by \texttt{path}$(\cdot,\cdot,\ell, \prec)$. 
    
    It thus remains to argue that \texttt{path}$(\cdot,\cdot,\ell, \prec)$ can be implemented using $\p(s,k)$ space. Here, we realize that the recursive invocation to \texttt{path}$(\cdot,\cdot,\ell-\ceil{\ell/2}, \prec)$ can reuse the space used by the recursive invocation to \texttt{path}$(\cdot,\cdot,\ceil{\ell/2}, \prec)$. Furthermore, each invocation needs to store only $\p(s,k)$ bits corresponding to the values of $C, D$ and $\ell$ it is invoked with before recursing. The recursion depth until the base case is reached is $\log(\ell) \le \log(N)=\p(s,k)$. Finally, the base case for $\ell=1$ requires checking if two DFAs $C$ and $D$ satisfy $C \prec D$, which, by assumption, requires $\p(s,k)$ space. In total, we obtain that \texttt{path}$(\cdot,\cdot,\ell, \prec)$ requires $\p(s,k)$ space.
\end{proof}

\popartialorder*
\begin{proof}
    Recall the definition of $\prec_2$:
    \begin{equation*}
        A \prec_2 A' \iff |L(A) \setminus L(A')| < \infty \text{ and } |L(A) \setminus L(A')| < |L(A') \setminus L(A)|.
    \end{equation*}
    We will argue that $\prec_2$ is irreflexive ($A \nprec_2 A$), asymmetric ($A \prec_2 B \implies B \nprec_2 A$) and transitive ($A \prec_2 B \text{ and } B \prec_2 C \implies A \prec_2 C$).
    \begin{enumerate}
        \item Irreflexivity: $A \nprec_2 A$, since $|L(A) \setminus L(A)|=0$, and $0 \nless 0$.
        \item Asymmetry: Suppose $A \prec_2 B$, which means that $|L(A) \setminus L(B)| < \infty$ and $|L(A) \setminus L(B)| < |L(B) \setminus L(A)|$. Assume for the sake of contradiction that $B \prec_2 A$. This means that $|L(B) \setminus L(A)| < \infty$ and $|L(B) \setminus L(A)| < |L(A) \setminus L(B)|$. But when both $|L(A) \setminus L(B)|$  and $|L(B) \setminus L(A)|$ are finite, $|L(A) \setminus L(B)| < |L(B) \setminus L(A)|$ and $|L(B) \setminus L(A)| < |L(A) \setminus L(B)|$ cannot simultaneously hold.
        \item Transitivity: Suppose $A \prec_2 B$ and $B \prec_2 C$; we want to show that $A \prec_2 C$. Since $A \prec_2 B$ and $B \prec_2 C$, we have that $|L(A) \setminus L(B)| < \infty$ and $|L(B) \setminus L(C)| < \infty$. Then,
        \begin{align*}
            &L(A) \setminus L(C) \subseteq (L(A) \setminus L(B)) \cup (L(B) \setminus L(C)) \\
            \implies\qquad& |L(A) \setminus L(C)| \le |L(A) \setminus L(B)| + |L(B) \setminus L(C)| < \infty. \tag{union bound}
        \end{align*}
        Thus, $|L(A)\setminus L(C)| < \infty$. Now, we wish to further show that $|L(A) \setminus L(C)| < |L(C) \setminus L(A)|$. Towards this, we will use the following key identity, which expresses $L(A) \setminus L(B)$ as a disjoint union:
        \begin{align*}
            &L(A) \setminus L(B) = \left(L(A) \setminus (L(B) \cup L(C))\right) \sqcup \left((L(A) \cap (L(C)) \setminus L(B)\right), \\
            \implies \qquad & |L(A) \setminus L(B)| = \left|L(A) \setminus (L(B) \cup L(C))\right| + \left|(L(A) \cap (L(C)) \setminus L(B)\right|.
        \end{align*}
        Expressing all remaining pairwise set differences in this fashion, we have
        \begin{align*}
            &|L(B) \setminus L(A)| = \left|L(B) \setminus (L(A) \cup L(C))\right| + \left|(L(B) \cap (L(C)) \setminus L(A)\right|, \\
            &|L(B) \setminus L(C)| = \left|L(B) \setminus (L(C) \cup L(A))\right| + \left|(L(B) \cap (L(A)) \setminus L(C)\right|, \\
            &|L(C) \setminus L(B)| = \left|L(C) \setminus (L(B) \cup L(A))\right| + \left|(L(C) \cap (L(A)) \setminus L(B)\right|, \\
            &|L(A) \setminus L(C)| = \left|L(A) \setminus (L(C) \cup L(B))\right| + \left|(L(A) \cap (L(B)) \setminus L(C)\right|, \\
            &|L(C) \setminus L(A)| = \left|L(C) \setminus (L(A) \cup L(B))\right| + \left|(L(C) \cap (L(B)) \setminus L(A)\right|.
        \end{align*}
        Since $A \prec_2 B$ and $B \prec_2 C$, we also have that $|L(A) \setminus L(B)|<|L(B) \setminus L(A)|$ and $|L(B) \setminus L(C)| < |L(C) \setminus L(B)|$. Using the above identities, this means that
        \begin{align}
            &\left|L(A) \setminus (L(B) \cup L(C))\right| + \left|(L(A) \cap (L(C)) \setminus L(B)\right| \nonumber \\
            &\hspace{3cm} < \left|L(B) \setminus (L(A) \cup L(C))\right| + \left|(L(B) \cap (L(C)) \setminus L(A)\right|, \label{eqn:set-diff-1}\\
            &\left|L(B) \setminus (L(C) \cup L(A))\right| + \left|(L(B) \cap (L(A)) \setminus L(C)\right| \nonumber \\
            &\hspace{3cm} < \left|L(C) \setminus (L(B) \cup L(A))\right| + \left|(L(C) \cap (L(A)) \setminus L(B)\right| \label{eqn:set-diff-2}.
        \end{align}
        All the quantities in the LHS of both \eqref{eqn:set-diff-1} and \eqref{eqn:set-diff-2} are finite, since both $|L(A) \setminus L(B)|$ and $|L(B) \setminus L(C)|$ are finite. Furthermore, this also implies that the term $\left|L(B) \setminus (L(A) \cup L(C))\right|$ in the RHS of \eqref{eqn:set-diff-1}, and the term $\left|(L(C) \cap (L(A)) \setminus L(B)\right|$ in the RHS of \eqref{eqn:set-diff-2} are both finite.
        
        Now recall: we wish to show that $|L(A) \setminus L(C)| < |L(C) \setminus L(A)|$. Since we have  shown that $|L(A) \setminus L(C)| < \infty$, if $|L(C) \setminus L(A)| = \infty$, then the inequality already holds. So, suppose that $|L(C) \setminus L(A)| < \infty$. Using the identity for $|L(C) \setminus L(A)|$ above, this means that both $\left|L(C) \setminus (L(A) \cup L(B))\right|$ and $\left|(L(C) \cap (L(B)) \setminus L(A)\right|$ are finite. But then, all the terms in the inequalities \eqref{eqn:set-diff-1} and \eqref{eqn:set-diff-2} are finite. Adding the two inequalities, and canceling the common terms on both sides, we get
        \begin{align}
            &\left|L(A) \setminus (L(C) \cup L(B))\right| + \left|L(A) \cap (L(B) \setminus L(C))\right| \nonumber \\
            &\hspace{3cm}< \left|L(C) \setminus (L(A) \cup L(B))\right| + \left|L(C) \cap (L(B) \setminus L(A))\right| \label{eqn:set-diff-3}.
        \end{align}
        But then, using the identities above one final time, this means that $|L(A) \setminus L(C)| < |L(C) \setminus L(A)|$, which is the desired inequality. Thus, we have shown that $A \prec_2 C$.
    \end{enumerate}
\end{proof}

\linearbound*

\begin{proof}
    The result essentially follows from Lemma 5 in \cite{gawrychowski2011minimising}; we flesh out a detailed proof for completeness.

    Consider a ``disjoint union automaton'' $M$, whose state space $Q_M=Q_A \sqcup Q_B$, so that $|Q_M|=|Q_A|+|Q_B| \le 2s$. In $M$, we retain transitions originally within states in $Q_A$ and $Q_B$, but there are no transitions across states in $Q_A$ and $Q_B$. That is, for any alphabet $\sigma \in \Sigma$, for any $q \in Q_A$, we have $\delta_M(q,\sigma)=\delta_A(q,\sigma)$, and for any $q \in Q_B$, we have $\delta_M(q, \sigma)=\delta_B(q,\sigma)$. $M$ has two distinct initial states $q_{0}^A$ and $q_0^B$, and accepting states $F = F_A \sqcup F_B$. We note that $M$ is not exactly a valid automaton, and we only construct it for the purpose of analysis.

    Let us the define the ``right language'' of a state $q \in M$ as follows:
    \begin{align*}
        L^+_M(q) = \{w \in \Sigma^*: \delta_M(q,w) \in F\}
    \end{align*}
    We immediately have that $L^+_M(q_0^A)=L(A)$ and $L^+_M(q_0^B)=L(B)$. Now, for any $p, q \in Q_M$, define the following (pseudo)distance $d_M(p,q)$:
    \begin{equation}
        \label{eqn:distance-def}
        d_M(p, q) = \min\{\ell \ge 0: L^+_M(p) \cap \Sigma^{\ge \ell} = L^+_M(q) \cap \Sigma^{\ge \ell}\}.
    \end{equation}
    where $\Sigma^{\ge \ell}=\{w \in \Sigma^*:|w| \ge \ell\}$. In words, $d_{M}(p,q)$ is the smallest $\ell$, such that the right languages $L^+_M(p)$ and $L^+_M(q)$, restricted to strings of length at least $\ell$, are equal. We can verify that $d_M(p, q)$ satisfies the following recursive formula:
    \begin{equation}
        \label{eqn:distance-recursive}
        d_M(p, q) = \begin{cases}
            0 & \text{if } L^+_M(p)=L^+_M(q) \\
            1 + \max_{\sigma \in \Sigma}\{d_M\left(\delta_M(p,\sigma),\delta_M(q,\sigma)\right)\} & \text{otherwise}.
        \end{cases}
    \end{equation}
    In fact, $d_M(p,q)$ is a so-called ``ultrapseudometric'' over the states in $Q_M$: we can immediately verify that $d_M(p,p)=0$, $d_M(p,q) \ge 0$ always and $d_{M}(p,q)=d_M(q,p)$; more importantly, $d_M(p,q)$ also satisfies the strong triangle inequality:
    \begin{align}
        \label{eqn:strong-triangle-inequality}
        d_M(p,q) \le \max(d_M(p,r), d_M(r,q)).
    \end{align}
    To see this, note that the inequality holds immediately if either $d_M(p,r)=\infty$ or $d_M(r,q)=\infty$. Otherwise, let $d_M(p,r)=a < \infty$ and $d_M(r,q)=b < \infty$. By definition of $d_M(\cdot,\cdot)$, we have that 
    \begin{align*}
        &L^+_M(p) \cap \Sigma^{\ge a} = L^+_M(r) \cap \Sigma^{\ge a} \quad \text{and} \quad L^+_M(r) \cap \Sigma^{\ge b} = L^+_M(q) \cap \Sigma^{\ge b} \\
        \implies\qquad&L^+_M(p)\cap \Sigma^{\ge \max(a,b)} = L^+_M(q) \cap \Sigma^{\ge \max(a,b)} \\
        \implies\qquad& d_M(p,q) \le \max(a,b) = \max(d_M(p,r), d_M(r,q)),
    \end{align*}
    as required.

    We now claim that, for any two states $p,q \in Q_M$, if $d_M(p,q) < \infty$, then $d_M(p,q) < |Q_M| \le 2s$. Before we go ahead and prove this, let us first see how it implies the lemma. For this, consider $d_M(q^A_0,q^B_0)$. Recall that $L^+_M(q_0^A)=L(A)$ and $L^+_M(q_0^B)=L(B)$, and by assumption, $|L(A) \triangle L(B)| < \infty$. But this necessarily means that $d_M(q^A_0,q^B_0) < \infty$ (otherwise, there would be arbitrarily long strings in the symmetric difference, making it an infinite set). So, by our claim, it holds that $d_M(q^A_0,q^B_0) \le 2s$. By definition of $d_M(\cdot, \cdot)$, this implies that any string in $\Sigma^*$ of length at least $2s-1$ is either in both of $L(A)$ and $L(B)$ or in neither, and hence proves the lemma.

    We now proceed to proving the promised claim: if $d_M(p,q) < \infty$, then $d_M(p,q) < |Q_M|$. Towards this, for any $i \ge 0$, consider the equivalence relation $D_i$ over the set $Q_M$ defined as follows:
    \begin{align}
        \label{eqn:relation-D_i-def}
        (p,q) \in D_i \iff d_M(p,q) \le i.
    \end{align}
    To see that this is an equivalence relation for every $i$, note that $(p,p) \in D_i$ always since $d_M(p,p)=0$ (reflexivity), and also, $(p,q) \in D_i \implies (q,p) \in D_i$ since $d_M(p,q)=d_M(q,p)$ (symmetry). For the transitive property, we use the strong triangle inequality \eqref{eqn:strong-triangle-inequality} from above: if $(p,r) \in D_i$ and $(r,q) \in D_i$, meaning that $d_M(p,r) \le i, d_M(r,q) \le i$, then
    \begin{align*}
        d_M(p,q) \le \max(d_M(p,r), d_M(r,q)) \le i,
    \end{align*}
    meaning that $(p,q) \in D_i$. 

    Now, let $n_i$ be the number of equivalence classes that $D_i$ partitions $Q_M$ into. We have that $n_0 \le |Q_M|$, and every $n_i \ge 1$. Observe now that by definition of $D_i$, it is also the case that $D_0 \subseteq D_1 \subseteq D_2 \subseteq \dots$. This immediately implies that
    \begin{align*}
        n_0 \ge n_1 \ge n_2 \dots.
    \end{align*}
    Since $n_0 \le |Q_M| < \infty$ and each $n_i \ge 1$, there must be some finite $i$ for which $n_i=n_{i+1}$. We now claim that if $n_i = n_{i+1}$ for any $i \ge 0$, then both: (1) for any two states $p,q \in Q_M$, $d_M(p,q) > i \implies d_M(p,q)=\infty$ , and (2) $n_j = n_i$ for every $j > i$.

    For (1), first note that $n_i = n_{i+1}$ implies that $D_i = D_{i+1}$. This follows simply because $D_i \subseteq D_{i+1}$, and both are equivalence relations on $Q_M$. Namely, if there is any pair $(p,q) \in D_{i+1}$ which is not in $D_i$, then this pair would merge two equivalence classes in $D_i$, contradicting $n_i = n_{i+1}$.

    Now, fix $p,q \in Q_M$, and suppose that $d_M(p,q) > i$. We will first argue that $d_M(p,q) \neq i+1$. Suppose not, meaning that $d_M(p,q) = i+1$. Then, by definition, $(p,q) \in D_{i+1}$. But $(p,q) \notin D_i$, contradicting that $D_i = D_{i+1}$. Thus, $d_M(p,q) \neq i+1$.
    
    Now, we argue that $d_M(p,q) \neq j$ for every finite $j > i+1$. Suppose not, meaning that $d_M(p,q) = j$ for some finite $j > i+1$. Since $j > i+1 \ge 1$, by the recursive definition in \eqref{eqn:distance-recursive},
    \begin{align*}
        j = d_M(p,q) = 1 + \max_{\sigma \in \Sigma}\{d_M\left(\delta_M(p,\sigma),\delta_M(q,\sigma)\right)\}.
    \end{align*}
    This means that for $\sigma \in \Sigma$ which is the maximizer above,
    \begin{align*}
        d_M\left(\delta_M(p,\sigma),\delta_M(q,\sigma)\right) = j-1.
    \end{align*}
    Denoting $\delta_M(p,\sigma) := p_1$ and $\delta_M(q,\sigma) := q_1$, we have that $\delta_M(p_1,q_1)=j-1$. Let $t=j-(i+1)$. Since $j$ is finite, $t$ is finite; so repeating the argument above $t$ times, we will obtain that
    \begin{align*}
        \delta_M(p_t, q_t) = j-t = i+1.
    \end{align*}
    But this means that $d_M(p_t,q_t) > i$, which, as we established above, necessarily means that $\delta_M(p_t, q_t) \neq i+1$. This is a contradiction, and hence $d_M(p,q) \neq j$. This proves (1).

    We now turn towards proving (2). Consider any $j > i$. For any two states $p,q \in Q_M$, if $d_M(p,q) \le j  < \infty$, then (1) above implies that $d_M(p,q) \le i$. We also trivially have that $d_M(p,q) \le i \implies d_M(p,q) \le j$. That is, for any two states $p,q \in Q_M$, it holds that
    \begin{align*}
        (p,q) \in D_i \iff d_M(p,q) \le i \iff d_M(p,q) \le j \iff (p,q) \in D_j.
    \end{align*}
    Thus, $D_i=D_j$ and hence $n_j=n_i$ for every $j > i$, which proves (2).

    Finally, putting (2) together with the fact that $|Q_M| \ge n_0 \ge n_1 \ge \dots$, where each $n_i \ge 1$, we necessarily have that by $i=|Q_M|-1$, it holds that $n_i = n_{i+1}$. But then the contrapositive of (1) above implies that for any $p,q \in Q_M$, $d_M(p,q) < \infty \implies d_M(p,q) \le |Q_M|-1$, which completes the entire proof.
\end{proof}

\SpaceEfficientGenerationTheorem*

\begin{proof}
The algorithm $\A$ runs \effiter$(\prec_2)$, which outputs, in the limit, a DFA $A_{i(t)}$ that satisfies $|L({A_{i(t)}}) \triangle K| < \infty$ (by Proposition \ref{prop:po2-converges-to-good-language}). Let $B$ be a DFA that accepts all strings of length at least $2s-1$, and rejects all strings of smaller length; $B$ can be constructed with $2s$ states $q_0,\dots,q_{2s-1}$ where every $q_i$ always transitions to $q_{i+1}$ for $0 \le i < 2s-1$. The initial state is $q_0$, every $q_i$ for $1 \le i < 2s-1$ is a rejecting sink state, and $q_{2s-1}$ is the only accepting sink state. By Lemma \ref{lemma:linear-bound-symmetric-diff}, we have that any $w \in L(B)$ is either in $L(A_{i(t)}) \cap K$ or in $\Sigma^* \setminus (L(A_{i(t)}) \cup K)$.

So, consider the DFA $A'$ which is the product DFA recognizing $L(A_{i(t)}) \cap L(B)$. The DFA $A'$ has at most $s \cdot 2s=2s^2$ states, and satisfies that $L(A') \subseteq L(A_{i(t)}) \cap K \subseteq K$. Furthermore, since any $w \in K \setminus A_{i(t)}$ must necessarily have length smaller than $2s-1$, $K \setminus A_{i(t)} \subseteq K \setminus L(B)$, and hence
\begin{align}
    K \setminus L(A') \subseteq (K \setminus A_{i(t)}) \cup (K \setminus L(B)) = K \setminus L(B). \label{eqn:breadth-equation}
\end{align}
Thus, $L(A')$ is only deficient of the strings in $K$ that have length at most $2s-2$. If $|\Sigma|=k=1$, we can have the algorithm $\A$ additionally keep track of a set $S$ of all the strings seen in the input that have length at most $2s-2$. Since there are at most $2s-1$ such strings (including the empty string), this requires only $\p(s)$ memory. Furthermore, since all the strings in $K$ are eventually revealed, the set $S$ eventually contains all strings in $K$ that have length at most $2s-2$. Finally, since $S$ can be recognized by a DFA having at most $2s-1$ states, the algorithm can output $\A(t)$ to be the product DFA that recognizes $L(A') \cup S$, which has at most $2s^2\cdot (2s-1)=O(s^3)$ states. In this case, we are guaranteed that $L(\A(t))=K$ eventually.

If $k > 1$, we simply have the algorithm output the DFA $\A(t)=A'$. In this case, \eqref{eqn:breadth-equation} gives that $|K \setminus L(A')| \le |K \setminus L(B)| = \frac{k^{2s-1}-1}{k-1} = O(k^{2s-2})$. 

\end{proof}
Let us revisit the trick we used above for the case where $|\Sigma|=1$; here, we had the algorithm keep track of the set $S$ of input strings that had length at most $2s-2$. With a unary alphabet, $S$ could have at most $O(s)$ strings, and hence the algorithm could store all of these in its $\p(s,k)$ memory budget. In fact, this allowed the algorithm to achieve the \textit{stronger} guarantee that $\A(t)=K$; i.e., $\A$ identified the target $K$ in the limit! When $k > 1$, storing the set $S$ requires $\exp(s,k)$ memory, which is why we simply returned $A'$ in this case. Nevertheless, we note that with $\exp(s,k)$ memory, we can use the same trick as in the $|\Sigma|=1$ case, and achieve identification in the limit.

%% file: Appendix_Sections/Appendix_Lower_Bounds.tex
\section{Deferred Proofs from \Cref{sec:lower-bounds}}\label{sec:appendix-lower-bounds}

\begin{proof}[Proof of Lemma~\ref{lem:index-lb}]
       Suppose that a deterministic protocol for $\textnormal{INDEX}_n$ uses $c < n\log_2k$ bits
    of communication. Then, there are at most $2^c < k^n$ possible messages that the 
    protocol can send, meaning
    that there exist two vectors $x, y \in [k]^n$, with $x \neq y$ on which the
    protocol sends the same message. Let $i$ be one coordinate in which $x$ and $y$ differ.
    In both instances of the problem (namely, $(x,i)$ and $(y,i)$), Bob receives the same message and  produces the same answer, even though $x_i \neq y_i$, which contradicts the correctness of the protocol.
\end{proof}

\paragraph{Relaxed symmetric-difference lower bound.}
We also spell out the proof of the second part of \Cref{thm:intro-main-lower-bound}, corresponding to the relaxed objective mentioned after \Cref{def:space-efficient-generation-in-the-limit}. In this formulation, the learner is not required to satisfy $L(\A(t)) \subseteq K$; instead, for every target $K \in \mathcal{C}_{s,k}$ and every surjective enumeration of $K$, the limiting hypothesis must satisfy
\[
    |L(\A(t)) \triangle K| \le \widetilde{\Delta}_{\mathcal{A}}(s,k)
\]
for all sufficiently large $t$.

We use exactly the same $\textnormal{INDEX}_n$ reduction and the same languages as in the proof of \Cref{lem:space-lower-bound-analytical}. Fix Bob's index $i$, and write $B_q := \Base(i,q)$ for $q \in \Sigma$. If Alice's input is $x$, then after Alice's prefix and Bob's continuation, the target language presented to the learner is $B_{x_i}$. Let $H$ be the limiting DFA output by the learner on this stream. Assume for contradiction that
\[
    \widetilde{\Delta}_{\mathcal{A}}(s,k) < k^\ell .
\]
Under the new definition, this means that
\[
    |L(H) \triangle B_{x_i}| < k^\ell .
\]
For any $q \neq x_i$, the two candidate target languages $B_q$ and $B_{x_i}$ differ exactly on the two disjoint blocks $L_{\map(i,q)}$ and $L_{\map(i,x_i)}$. Therefore,
\[
    |B_q \triangle B_{x_i}| = 2k^\ell .
\]
By the triangle inequality for symmetric difference,
\[
    |L(H) \triangle B_q|
    \ge |B_q \triangle B_{x_i}| - |L(H) \triangle B_{x_i}|
    > 2k^\ell - k^\ell
    = k^\ell .
\]
On the other hand, $|L(H) \triangle B_{x_i}| < k^\ell$. Hence $q = x_i$ is the unique minimizer of $|L(H) \triangle \Base(i,q)|$
over $q \in \Sigma$.

Bob can therefore decode $x_i$ by computing, for every $q \in \Sigma$, the size of the regular language $L(H) \triangle \Base(i,q)$, and outputting the unique minimizer. Concretely, Bob constructs the product DFA for this symmetric difference and counts its accepted language when it is finite; an infinite value may be treated as $+\infty$. Under the promise above, the true candidate is finite and uniquely smallest. This gives the same one-way $\textnormal{INDEX}_n$ protocol as in \Cref{lem:space-lower-bound-analytical} with Alice's message equal to the memory state of the learner. The $\textnormal{INDEX}_n$ lower bound is contradicted whenever that memory state has fewer than $n\log_2 k$ bits. Thus $\widetilde{\Delta}_{\mathcal{A}}(s,k) \ge k^\ell$, and substituting the same values of $p,\ell,n$ as in \Cref{lem:space-lower-bound-analytical} gives the same quantitative lower bound.

\paragraph{Further remarks for the Proof of \Cref{lem:space-lower-bound-analytical}.}
As we can see in the proof of $\textnormal{INDEX}_n$ above, the lower bound is purely information-theoretic
and contains no requirements about the computation power of Alice and Bob. This means that even the existence
of two functions, an encoding function of $x$ to $\ell < n \log_2 k$ bits, let it be $f: [k]^n \to [2]^\ell$, and a decoding function $g: [n] \times [2]^\ell \to [k]$ such that for all $x \in [k]^n$ and all $i \in [n]$ it holds that $g(i, f(x)) = x_i$, violates the lower bound. In our construction, you can view $f(x)$ as the memory state $\sigma$
of the learner after Alice has given it her part of the input. Now, let $t^*(i,q)$ be the timestep after which the learner $\mathcal{A}$ has converged to its limiting DFA, on the enumeration of our construction when the target language is $\Base(i,q)$. Let also $T^*$ be the maximum over $t^*(i,q)$ for all $i \in [n]$ and $q \in [k]$.
Then, the decoding function $g$ can be defined as the output of Bob when he waits for $(T^*+1)$ time steps and then performs the corresponding finite DFA test. Of course, the existence of $g$ contradicts the $\textnormal{INDEX}_n$ lower bound.

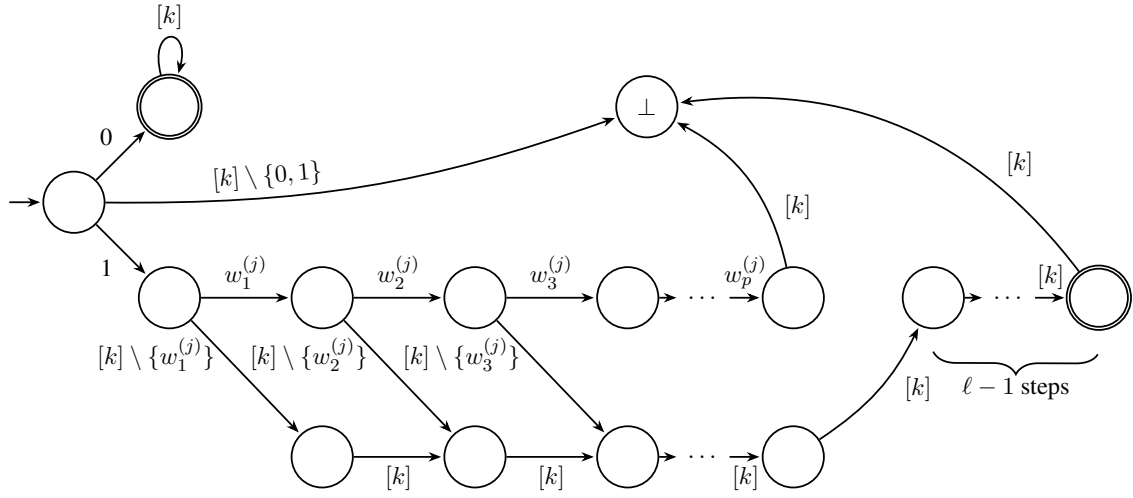
\begin{figure}[h]
    \centering
    \resizebox{\textwidth}{!}{
\begin{tikzpicture}[
    shorten >=1pt, 
    node distance=1.8cm, 
    on grid, 
    auto, 
    thick,
    initial text={},
    >={Stealth[length=2mm]}
]

    % --- States: Initial & Top ---
    \node[state, initial] (q0) {};
    \node[state, accepting, above right=1.5cm and 1.5cm of q0] (q1) {};
    % Dead state placed in upper-middle to allow back-arcs
    \node[state, right=7.5cm of q1] (dead) {$\perp$};
    
    % --- States: Middle Row (w sequence) ---
    \node[state, below right=1.5cm and 1.5cm of q0] (w0) {};
    \node[state, right=2.4cm of w0] (w1) {};
    \node[state, right=2.4cm of w1] (w2) {};
    \node[state, right=2.4cm of w2] (w3) {};
    \node[draw=none, right=1.2cm of w3] (dots_top) {$\dots$};
    \node[state, right=1.4cm of dots_top] (wp) {};
    
    % --- States: Bottom Row (b sequence / Sink) ---
    \node[state, below=2.5cm of w1] (b1) {};
    \node[state, below=2.5cm of w2] (b2) {};
    \node[state, below=2.5cm of w3] (b3) {};
    \node[draw=none, below=2.5cm of dots_top] (dots_bot) {$\dots$};
    \node[state, below=2.5cm of wp] (bp) {};

    % --- States: Success Sequence (l steps) ---
    \node[state, right=2.2cm of wp] (l1) {};
    \node[draw=none, right=1.2cm of l1] (dots_end) {$\dots$};
    \node[state, accepting, right=1.4cm of dots_end] (q_acc) {};

    % --- Transitions: Initial Split ---
    \path[->] 
        (q0) edge node {0} (q1)
        (q0) edge node [swap] {1} (w0)
        (q0) edge [bend right=10] node [pos=0.3, sloped, above] {$[k] \setminus \{0, 1\}$} (dead)
        (q1) edge [loop above] node {$[k]$} (q1);

    % --- Transitions: Top Row & Diagonals ---
    \path[->] 
        (w0) edge node {$w_1^{(j)}$} (w1)
        (w1) edge node {$w_2^{(j)}$} (w2)
        (w2) edge node {$w_3^{(j)}$} (w3)
        (w3) edge node {} (dots_top)
        (dots_top) edge node {$w_p^{(j)}$} (wp)
        
        % NEW: wp now goes to dead state
        (wp) edge [bend right=25] node [swap, pos=0.2] {$[k]$} (dead)
        
        % Diagonals (Labels higher via pos=0.3)
        (w0) edge node [pos=0.3, left] {$[k] \setminus \{w_1^{(j)}\}$} (b1)
        (w1) edge node [pos=0.3, left] {$[k] \setminus \{w_2^{(j)}\}$} (b2)
        (w2) edge node [pos=0.3, left] {$[k] \setminus \{w_3^{(j)}\}$} (b3);
        % (dots_top) edge node [pos=0.3, left] {$[k] \setminus w_{p-1}^{(j)}$} (bp);

    % --- Transitions: Bottom Row & Success Path ---
    \path[->]
        % Forward arcs on bottom (Labels below via swap)
        (b1) edge node [swap] {$[k]$} (b2)
        (b2) edge node [swap] {$[k]$} (b3)
        (b3) edge node {} (dots_bot)
        (dots_bot) edge node [swap] {$[k]$} (bp)
        
        % NEW: bp now goes to l1 (state after wp)
        (bp) edge [bend right=15] node [swap, pos=0.7] {$[k]$} (l1)
        
        % Final sequence transitions
        (l1) edge node {} (dots_end)
        (dots_end) edge node {$[k]$} (q_acc)
        (q_acc) edge [bend right=30] node [swap, near start] {$[k]$} (dead);

    % --- Brace for l steps ---
    \draw [decorate, decoration={brace, amplitude=10pt, mirror, raise=10pt}]
        (l1.south) -- (q_acc.south) 
        node [midway, yshift=-35pt] {$\ell-1$  steps};

\end{tikzpicture}
}
    \caption{A Deterministic Finite Automaton (DFA) accepting $\textnormal{Base}(j,q)$ for some $j \in [n]$ and $q \in \Sigma$. We use the notation $\textnormal{enc}_k(j) = w_1^{(j)} \circ w_2^{(j)} \circ w_3^{(j)} \circ \cdots \circ w_p^{(j)}$}
    \label{fig:dfa_accepting_base}
\end{figure}

%% file: references.bib
@misc{kleinberg2026languagegenerationlimitbounded,
      title={On Language Generation in the Limit with Bounded Memory}, 
      author={Jon Kleinberg and Anay Mehrotra and Amin Saberi and Grigoris Velegkas},
      year={2026},
      eprint={2605.30324},
      archivePrefix={arXiv},
      primaryClass={cs.DS}, 
}

@article{kleinberg2024language,
  title={Language generation in the limit},
  author={Kleinberg, Jon and Mullainathan, Sendhil},
  journal={Advances in Neural Information Processing Systems},
  volume={37},
  pages={66058--66079},
  year={2024}
}

@article{kalavasis2024characterizations,
  title={Characterizations of language generation with breadth},
  author={Kalavasis, Alkis and Mehrotra, Anay and Velegkas, Grigoris},
  journal={arXiv preprint arXiv:2412.18530},
  pages={3},
  year={2024}
}

@article{papazov2025learning,
  title={Learning Algorithms in the Limit},
  author={Papazov, Hristo and Flammarion, Nicolas},
  journal={arXiv preprint arXiv:2506.15543},
  year={2025}
}

@article{charikar2024exploring,
  title={Exploring facets of language generation in the limit},
  author={Charikar, Moses and Pabbaraju, Chirag},
  journal={arXiv preprint arXiv:2411.15364},
  year={2024}
}

@article{kleinberg2025density,
  title={Density Measures for Language Generation},
  author={Kleinberg, Jon and Wei, Fan},
  journal={arXiv preprint arXiv:2504.14370},
  year={2025}
}

@article{li2024generation,
  title={Generation through the lens of learning theory},
  author={Li, Jiaxun and Raman, Vinod and Tewari, Ambuj},
  journal={arXiv preprint arXiv:2410.13714},
  year={2024}
}

@article{charikar2025pareto,
  title={Pareto-optimal Non-uniform Language Generation},
  author={Charikar, Moses and Pabbaraju, Chirag},
  journal={arXiv preprint arXiv:2510.02795},
  year={2025}
}

@article{raman2025generation,
  title={Generation from Noisy Examples},
  author={Raman, Ananth and Raman, Vinod},
  journal={arXiv preprint arXiv:2501.04179},
  year={2025}
}

@article{peale2025representative,
  title={Representative Language Generation},
  author={Peale, Charlotte and Raman, Vinod and Reingold, Omer},
  journal={arXiv preprint arXiv:2505.21819},
  year={2025}
}

@article{hanneke2025union,
  title={On Union-Closedness of Language Generation},
  author={Hanneke, Steve and Karbasi, Amin and Mehrotra, Anay and Velegkas, Grigoris},
  journal={arXiv preprint arXiv:2506.18642},
  year={2025}
}

@article{charikar2025characterization,
  title={A Characterization of List Language Identification in the Limit},
  author={Charikar, Moses and Pabbaraju, Chirag and Tewari, Ambuj},
  journal={arXiv preprint arXiv:2511.04103},
  year={2025}
}

@article{angluin1980inductive,
  title={Inductive inference of formal languages from positive data},
  author={Angluin, Dana},
  journal={Information and control},
  volume={45},
  number={2},
  pages={117--135},
  year={1980},
  publisher={Elsevier}
}

@article{gold1967language,
  title={Language identification in the limit},
  author={Gold, E Mark},
  journal={Information and control},
  volume={10},
  number={5},
  pages={447--474},
  year={1967},
  publisher={Elsevier}
}

@book{esparza2023automata,
  title={Automata theory: An algorithmic approach},
  author={Esparza, Javier and Blondin, Michael},
  year={2023},
  publisher={MIT Press}
}

@article{sipser1996introduction,
  title={Introduction to the Theory of Computation},
  author={Sipser, Michael},
  journal={ACM Sigact News},
  volume={27},
  number={1},
  pages={27--29},
  year={1996},
  publisher={ACM New York, NY, USA}
}

@article{hopcroft2001introduction,
  title={Introduction to automata theory, languages, and computation},
  author={Hopcroft, John E and Motwani, Rajeev and Ullman, Jeffrey D},
  journal={Acm Sigact News},
  volume={32},
  number={1},
  pages={60--65},
  year={2001},
  publisher={ACM New York, NY, USA}
}

@incollection{hopcroft1971n,
  title={An n log n algorithm for minimizing states in a finite automaton},
  author={Hopcroft, John},
  booktitle={Theory of machines and computations},
  pages={189--196},
  year={1971},
  publisher={Elsevier}
}

@article{gries1973describing,
  title={Describing an algorithm by Hopcroft},
  author={Gries, David},
  journal={Acta Informatica},
  volume={2},
  number={2},
  pages={97--109},
  year={1973},
  publisher={Springer}
}

@article{badr2009hyper_a,
  title={Hyper-minimizing minimized deterministic finite state automata},
  author={Badr, Andrew and Geffert, Viliam and Shipman, Ian},
  journal={RAIRO-Theoretical Informatics and Applications},
  volume={43},
  number={1},
  pages={69--94},
  year={2009},
  publisher={EDP Sciences}
}

@inproceedings{gawrychowski2011minimising,
  title={On minimising automata with errors},
  author={Gawrychowski, Pawel and Jez, Artur and Maletti, Andreas},
  booktitle={International Symposium on Mathematical Foundations of Computer Science},
  pages={327--338},
  year={2011},
  organization={Springer}
}

@article{maletti2011optimal,
  title={Optimal hyper-minimization},
  author={Maletti, Andreas and Quernheim, Daniel},
  journal={International Journal of Foundations of Computer Science},
  volume={22},
  number={08},
  pages={1877--1891},
  year={2011},
  publisher={World Scientific}
}

@article{chomsky1956three,
  title={Three models for the description of language},
  author={Chomsky, Noam},
  journal={IRE Transactions on information theory},
  volume={2},
  number={3},
  pages={113--124},
  year={1956},
  publisher={IEEE}
}

@article{miller1963finitary,
  title={Finitary models of language users},
  author={Miller, George A and Chomsky, Noam},
  year={1963}
}

@article{marcus1993negative,
  title={Negative evidence in language acquisition},
  author={Marcus, Gary F},
  journal={Cognition},
  volume={46},
  number={1},
  pages={53--85},
  year={1993},
  publisher={Elsevier}
}

@article{chouinard2003adult,
  title={Adult reformulations of child errors as negative evidence},
  author={Chouinard, Michelle M and Clark, Eve V},
  journal={Journal of child language},
  volume={30},
  number={3},
  pages={637--669},
  year={2003},
  publisher={Cambridge University Press}
}

@article{svete2023recurrent,
  title={Recurrent neural language models as probabilistic finite-state automata},
  author={Svete, Anej and Cotterell, Ryan},
  journal={arXiv preprint arXiv:2310.05161},
  year={2023}
}

@article{mahowald2024dissociating,
  title={Dissociating language and thought in large language models},
  author={Mahowald, Kyle and Ivanova, Anna A and Blank, Idan A and Kanwisher, Nancy and Tenenbaum, Joshua B and Fedorenko, Evelina},
  journal={Trends in cognitive sciences},
  volume={28},
  number={6},
  pages={517--540},
  year={2024},
  publisher={Elsevier}
}

@article{brown2020language,
  title={Language models are few-shot learners},
  author={Brown, Tom and Mann, Benjamin and Ryder, Nick and Subbiah, Melanie and Kaplan, Jared D and Dhariwal, Prafulla and Neelakantan, Arvind and Shyam, Pranav and Sastry, Girish and Askell, Amanda and others},
  journal={Advances in neural information processing systems},
  volume={33},
  pages={1877--1901},
  year={2020}
}

@article{radford2019language,
  title={Language models are unsupervised multitask learners},
  author={Radford, Alec and Wu, Jeffrey and Child, Rewon and Luan, David and Amodei, Dario and Sutskever, Ilya and others},
  journal={OpenAI blog},
  volume={1},
  number={8},
  pages={9},
  year={2019}
}

@article{savitch1970relationships,
  title={Relationships between nondeterministic and deterministic tape complexities},
  author={Savitch, Walter J},
  journal={Journal of computer and system sciences},
  volume={4},
  number={2},
  pages={177--192},
  year={1970},
  publisher={Elsevier}
}

@article{merrill2024illusion,
  title={The illusion of state in state-space models},
  author={Merrill, William and Petty, Jackson and Sabharwal, Ashish},
  journal={arXiv preprint arXiv:2404.08819},
  year={2024}
}

@article{svete2024transformers,
  title={Transformers Can Represent $ n $-gram Language Models},
  author={Svete, Anej and Cotterell, Ryan},
  journal={arXiv preprint arXiv:2404.14994},
  year={2024}
}

@InCollection{sep-church-turing,
    author       =	{Copeland, B. Jack},
    title        =	{{The Church-Turing Thesis}},
    booktitle    =	{The {Stanford} Encyclopedia of Philosophy},
    editor       =	{Edward N. Zalta and Uri Nodelman},
    year         =	{2026},
    edition      =	{{S}pring 2026},
    publisher    =	{Metaphysics Research Lab, Stanford University}
    }

@article{valiant1984theory,
  title={A theory of the learnable},
  author={Valiant, Leslie G},
  journal={Communications of the ACM},
  volume={27},
  number={11},
  pages={1134--1142},
  year={1984},
  publisher={ACM New York, NY, USA}
}

@book{kearns1994introduction,
  title={An introduction to computational learning theory},
  author={Kearns, Michael J and Vazirani, Umesh},
  year={1994},
  publisher={MIT press}
}

@book{aggarwal2007data,
  title={Data streams: models and algorithms},
  author={Aggarwal, Charu C},
  volume={31},
  year={2007},
  publisher={Springer Science \& Business Media}
}

@inproceedings{bar-yossef2002coco,
  author       = {Ziv Bar{-}Yossef and
                  T. S. Jayram and
                  Ravi Kumar and
                  D. Sivakumar},
  title        = {Information Theory Methods in Communication Complexity},
  booktitle    = {Proceedings of the 17th Annual {IEEE} Conference on Computational
                  Complexity},
  pages        = {93--102},
  publisher    = {{IEEE} Computer Society},
  year         = {2002},
  doi          = {10.1109/CCC.2002.1004344}
}

@article{jayram2008toc,
  author       = {T. S. Jayram and
                  Ravi Kumar and
                  D. Sivakumar},
  title        = {The One-Way Communication Complexity of Hamming Distance},
  journal      = {Theory Comput.},
  volume       = {4},
  number       = {1},
  pages        = {129--135},
  year         = {2008},
  doi          = {10.4086/TOC.2008.V004A006}
}

@incollection{kushilevitz1997communication,
  title={Communication complexity},
  author={Kushilevitz, Eyal},
  booktitle={Advances in Computers},
  volume={44},
  pages={331--360},
  year={1997},
  publisher={Elsevier}
}
